\pdfoutput=1
\RequirePackage[l2tabu, orthodox]{nag}

\documentclass[paper=a4, fontsize=11pt,final]{scrartcl}

\usepackage[english]{babel}
\usepackage[utf8]{inputenc}
\usepackage{cmap}
\usepackage[T1]{fontenc}
\usepackage{textcomp}
\usepackage{lmodern}
\usepackage{xspace}
\usepackage{amsmath}

\makeatletter
\@ifpackageloaded{hyperref}{
\hypersetup{
    final,
    colorlinks,
    linkcolor={red!80!black},
    citecolor={blue!50!black},
    urlcolor={blue!80!black},
    pdftitle = {On the Complexity of Linear Temporal Logic with Team Semantics},
    pdfauthor = {Martin Lück}
}
}
{
\usepackage{hyperref}
}
\makeatother

\usepackage{cleveref}
\usepackage{aliascnt}

\usepackage[heavycircles]{stmaryrd}

\usepackage{amssymb}

\usepackage{authblk}

\usepackage{bookmark}

\usepackage{pdfcolmk}

\usepackage{braket}

\usepackage[babel,autostyle=true,strict]{csquotes}
\MakeOuterQuote{"}

\usepackage{fixmath}

\usepackage[fixamsmath,disallowspaces]{mathtools}

\usepackage{stackengine}

\usepackage[protrusion=true,expansion=true,tracking=true,kerning=true,spacing=true]{microtype}
\usepackage[xspace]{ellipsis}

\usepackage{enumitem}

\usepackage{booktabs}
\usepackage{array}
\newcolumntype{R}{>{$}r<{$}}
\newcolumntype{L}{>{$}l<{$}}
\newcolumntype{C}{>{$}c<{$}}

\usepackage{colortbl}

\usepackage{makecell}

\usepackage{amsthm}

\theoremstyle{plain}
\newtheorem{theorem}{Theorem}[section]

\newaliascnt{lemma}{theorem}
\newaliascnt{corollary}{theorem}
\newaliascnt{definition}{theorem}
\newaliascnt{example}{theorem}
\newaliascnt{proposition}{theorem}

\theoremstyle{definition}
\newtheorem{definition}[definition]{Definition}
\theoremstyle{plain}
\newtheorem{proposition}[proposition]{Proposition}

\newcommand*{\altqed}{\hfill\small\ensuremath{\triangleleft}}

\newtheorem{lemma}[lemma]{Lemma}
\newtheorem{corollary}[corollary]{Corollary}

\aliascntresetthe{lemma}
\aliascntresetthe{corollary}
\aliascntresetthe{definition}
\aliascntresetthe{example}
\aliascntresetthe{proposition}

\usepackage{tikz}
\usetikzlibrary{arrows}
\usetikzlibrary{positioning}
\usetikzlibrary{backgrounds}
\usetikzlibrary{patterns,decorations.pathreplacing}
\usetikzlibrary{fit}

\usepackage[backend=biber,style=numeric,sortcites,url=false,doi=false,maxbibnames=31,maxcitenames=2,citetracker=true]{biblatex}

\DeclareFieldFormat[article]{title}{\mkbibemph{#1}}
\DeclareFieldFormat[article]{journal}{#1}
\DeclareFieldFormat[article]{journaltitle}{#1}
\DeclareFieldFormat[article]{volume}{\textbf{#1}}
\DeclareFieldFormat[article]{number}{no.~#1}

\DeclareFieldFormat[incollection]{title}{\mkbibemph{#1}}
\DeclareFieldFormat[incollection]{booktitle}{#1}
\DeclareFieldFormat[incollection]{volume}{\textbf{#1}}
\DeclareFieldFormat[incollection]{number}{(#1)}

\DeclareFieldFormat[inproceedings]{title}{\mkbibemph{#1}}
\DeclareFieldFormat[inproceedings]{booktitle}{#1}
\DeclareFieldFormat[inproceedings]{volume}{\textbf{#1}}
\DeclareFieldFormat[inproceedings]{number}{(#1)}

\DeclareFieldFormat[unpublished]{title}{\mkbibemph{#1}}

\DeclareFieldFormat[thesis]{title}{\mkbibemph{#1}}

\renewbibmacro*{volume+number+eid}{\printfield{volume}\space \printtext[parens]{\usebibmacro{date}}
  \setunit{\addcomma\space}
  \printfield{number}\setunit{\addcomma\space}\printfield{eid}
  }

\renewbibmacro*{issue+date}{\iffieldundef{issue}
      {}
      {\printfield{issue}}\newunit}

\renewbibmacro{in:}{}

\DeclareMathOperator{\ran}{\mathsf{ran}}

\newcommand*{\ie}{i.e.\@\xspace}
\newcommand*{\eg}{e.g.\@\xspace}
\newcommand*{\wrt}{w.\,r.\,t.\@\xspace}
\newcommand*{\Wloss}{W.l.o.g.\@\xspace}
\newcommand*{\wloss}{w.l.o.g.\@\xspace}
\newcommand*{\logic}[1]{\ensuremath{\mathsf{#1}}\xspace}

\newcommand*{\sneg}{{\sim}}

\newcommand*{\depopp}{\mathsf{dep}}
\newcommand*{\dep}[1]{{\text{\textsf{dep}}{\ensuremath(#1)}}}
\newcommand*{\MSO}{\logic{MSO}}
\newcommand*{\LTL}{\logic{LTL}}
\newcommand*{\AP}{\ensuremath\mathsf{AP}}
\providecommand{\dfn}{\mathrel{\mathop:}=}
\providecommand{\ddfn}{\mathrel{\mathop{{\mathop:}{\mathop:}}}=}

\newcommand*{\logicOpFont}[1]{\mathsf{#1}}

\newcommand*{\complClFont}[1]{\mathbf{#1}}
\newcommand*{\td}{\mathrm{td}}
\newcommand*{\X}{\protect\ensuremath\logicOpFont{X}}
\newcommand*{\F}{\protect\ensuremath\logicOpFont{F}}
\newcommand*{\G}{\protect\ensuremath\logicOpFont{G}}

\newcommand*{\U}{\protect\ensuremath\logicOpFont{U}}
\newcommand*{\R}{\protect\ensuremath\logicOpFont{R}}
\newcommand*{\imp}{\rightarrow}
\newcommand*{\equi}{\leftrightarrow}
\newcommand*{\N}{\protect\ensuremath{\mathbb{N}}\xspace}
\newcommand*{\size}[1]{{\ensuremath{\vert\nobreak#1\nobreak\vert}}}
\newcommand*{\LR}{\Leftrightarrow}
\renewcommand*{\phi}{\varphi}
\newcommand*{\restr}[2]{#1{\upharpoonright}#2}
\newcommand*{\calA}{\protect\ensuremath{\mathcal{A}}}
\newcommand*{\calC}{\protect\ensuremath{\mathcal{C}}}

\newcommand*{\calI}{\protect\ensuremath{\mathcal{I}}}

\newcommand*{\calK}{\protect\ensuremath{\mathcal{K}}}

\newcommand*{\calT}{\protect\ensuremath{\mathcal{T}}}

\newcommand*{\calV}{\protect\ensuremath{\mathcal{V}}}

\newcommand*{\vroot}{\mathsf{root}}
\newcommand*{\vend}{\mathsf{end}}
\newcommand*{\ulc}{\mathrm{ulc}}
\newcommand*{\ulp}{\mathrm{ulp}}
\newcommand*{\fraka}{\mathfrak{a}}
\newcommand*{\frakb}{\mathfrak{b}}
\newcommand*{\frakc}{\mathfrak{c}}

\newcommand*{\frakT}{\mathfrak{T}}

\newcommand*{\NEXPTIME}{\protect\ensuremath{\complClFont{NEXPTIME}}\xspace}

\newcommand*{\PSPACE}{\protect\ensuremath{\complClFont{PSPACE}}\xspace}
\newcommand*{\NP}{\protect\ensuremath{\complClFont{NP}}\xspace}

\newcommand*{\qa}[1]{\forall^{1}_{#1}}
\newcommand*{\qaa}[1]{\forall^{\subseteq}_{#1}}
\newcommand*{\qe}[1]{\exists^1_{#1}}
\newcommand*{\qee}[1]{\exists^{\subseteq}_{#1}}
\newcommand*{\hook}{\hookrightarrow}

\newcommand*{\timp}{\rightarrowtriangle}
\newcommand*{\tequiv}{\leftrightarrowtriangle}
\newcommand*{\steq}{\equiv_{\mathrm{st}}}
 
\author{\large Martin Lück\\\large Leibniz Universität Hannover, Germany\\ \large \texttt{lueck@thi.uni-hannover.de}}
\date{\vspace{-5ex}}

\title{On the Complexity of Linear Temporal Logic with Team Semantics}

\bibliography{main}

\begin{document}

\maketitle

\begin{abstract}
\textbf{Abstract.}
A specification given as a formula in linear temporal logic (LTL) defines a system by its set of traces.
However, certain features such as information flow security constraints are rather modeled as so-called hyperproperties, which are sets of sets of traces.
One logical approach to this is team logic, which is a logical framework for the specification of dependence and independence of information.
LTL with team semantics has recently been discovered as a logic for hyperproperties.

We study the complexity theoretic aspects of LTL with so-called synchronous team semantics and Boolean negation, and prove that both its model checking and satisfiability problems are highly undecidable, and equivalent to the decision problem of third-order arithmetic.
Furthermore, we prove that this complexity already appears at small temporal depth and with only the "future" modality $\F$.
Finally, we also introduce a team-semantical generalization of stutter-invariance.
\end{abstract}



\section{Introduction}

Linear temporal logic (LTL) \cite{Pnueli77} is a successful specification language for state-based transition systems, with many applications in software and hardware verification and protocol design \cite{BK08,Schneider04}.
This is not limited to correctness of the program output with respect to the input, but also covers aspects such as fairness, safety and deadlock-freeness.

The execution behaviour of a system is usually modeled as a trace, that is, an infinite sequence of propositional assignments.
However, certain important properties of systems are not expressible as classes of traces.
One simple example is \emph{bounded termination} of the system, that is, all traces reach some final state after at most $c$ steps, where $c$ does depend on the system not but on the trace.
Another is \emph{observational determinism}, also \emph{noninterference}, which means that the externally visible development of a trace should only depend on the input from the view of an external observer.

Recent research on temporal logic has also accounted for such properties of systems as the above.
Clarkson and Schneider~\cite{CS08} coined the name \emph{hyperproperties} for these, which are formally \emph{sets of sets} of traces.
A logic to specify hyperproperties, called HyperLTL, was introduced by Clarkson et al.~\cite{CFKMRS14}.
It extends LTL by universal and existential trace quantifiers, and can for example express observational determinism and has decidable model checking.
On the other hand, its satisfiability problem is undecidable and it still lacks the expressive power for properties such as bounded termination~\cite{CFKMRS14}.

\smallskip

Another novel logical framework useful for expressing hyperproperties is \emph{team semantics}~\cite{KMV018}.
Generally speaking, team semantics extends classical logic such that formulas are evaluated not over single assignments, states, etc., but instead over sets of those.
It has been studied in the context of first-order logic, propositional logic and modal logic, for example~\cite{Vaa07,Vaa08,VY16}.

Team semantics has also been adapted to temporal logics like CTL~\cite{KMV15} and recently LTL~\cite{KMV018} by Krebs et al.
They distinguish two kinds of team semantics for LTL: \emph{asynchronous} semantics, which is strictly weaker than HyperLTL, and \emph{synchronous} semantics, which is expressively incomparable to HyperLTL.
Roughly speaking, the future operator $\F$ either "asynchronously" skips a finite amount of time on every trace that is unbounded and depends on the trace, or "synchronously" advances to a point in the future that is equal for all traces, and similarly for the other temporal operators.

In team logic, often non-classical atoms are added to the syntax that express properties of teams.
The prime example is the \emph{dependence atom} $\dep{p_1,\ldots,p_n;q}$, which was introduced by Väänänen~\cite{Vaa07} in the context of team semantics.
It states that the value of $q$ functionally depends on the values of $p_1,\ldots,p_n$.
Applied to traces, this defines a hyperproperty.
In fact, the dependence atom allows to intuitively implement features like observational determinism, for example the formula
\begin{align*}
  \bigwedge_{i = 1}^k\dep{\mathsf{in}_1,\ldots,\mathsf{in}_n;\G(\vend\to \mathsf{out}_i)}
\end{align*}
states that the truth of the formulas $\G(\vend\to \mathsf{out}_i)$, and thus the values of $\mathsf{out}_1,\ldots,\mathsf{out}_k$ at the end of the computation on each trace, depend only on the values of the propositions $\mathsf{in}_1,\ldots,\mathsf{in}_n$ at the beginning of that trace.

Several other non-classical atoms have been considered in the past, \eg, independence, inclusion, and exclusion constraints as well as atoms for counting.
We refer the reader to the literature \cite{0037838,Kontinen13} for further information.
(Our results hold regardless of whether those atoms are available or not.)

Team semantics generally lacks the Boolean negation, and although the connective $\neg$ is part of the syntax, it is not the classical negation and does not satisfy for example the law of excluded middle~\cite{Vaa07}.
Sometimes Boolean negation, also called contradictory negation or strong negation, is re-introduced and denoted by $\sneg$.
This usually increases the expressive power greatly, as was shown for propositional, modal and first-order team logic~\cite{Vaa07,VY17,KMSV15}.
On the other hand, this is not surprising, as in the context of proposition-based logics, the non-classical atoms of dependency, independence, inclusion and others are all expressible by polynomial-sized formulas if $\sneg$ is available~\cite{LV19}.

\smallskip

The most important logical decision problems are model checking (does the given system satisfy the given formula?) and satisfiability (is the given formula true in some system?).
Not only model checking, but also satisfiability is routinely solved in practice, since an unsatisfiable and hence self-contradictory specification is of no use.
Thus checking for satisfiability is a sensible heuristic to avoid errors in the specification.~\cite{BK08}

Model checking and the satisfiability problem sometimes become easier for fragments of the logic.
For example, while both problems are well-known to be $\PSPACE$-complete for LTL~\cite{SC85}, their complexity can drop down to $\NP$ or even less when the input formulas are restricted, for example if they use only certain subsets of temporal operators, if the nesting depth of temporal operators is small, or if the number of distinct propositional symbols is bounded~\cite{SC85,DS02,Sch02}.
For HyperLTL, some research in the same direction has been done, with focus on fragments of small alternation depth of trace quantifiers and small temporal depth~\cite{abs-1907-05070,CFKMRS14,FH16}.

\smallskip

\textbf{Contributions.}
We consider the logic $\LTL(\sneg)$, which in this paper denotes linear temporal logic (LTL) with team semantics, with Boolean negation $\sneg$, and with synchronous semantics of the temporal connectives.

We begin with identifying a \emph{stutter-invariant} fragment of $\LTL(\sneg)$.
Classically, a formula $\phi$ being stutter-invariant means that finitely often repeating arbitrary labels on a trace does not change the truth of $\phi$.
Peled and Wilke~\cite{PW97} showed that the stutter-invariant LTL-definable trace properties are exactly those that are definable without the "nexttime" operator $\X$.
One application of stutter-invariant formulas is the hierarchical refinement of specifications and structures, where atomic transitions are replaced by a complex subroutine~\cite{lamport1983,PW97}.

In \Cref{sec:stutter}, we generalize the notion of stutter-equivalence to teams, and thus to hyperproperties, and partially obtain a similar result: we show that every $\X$-free formula is stutter-invariant.
Whether the converse holds like in the classical case is left open.

Afterwards, we turn to the computational complexity of $\LTL(\sneg)$.
It turns out that both its satisfiability and model checking problem are equivalent to third-order arithmetic.
In particular, model checking and satisfiability have the same complexity, which is well-known feature of LTL that now recurs in team semantics.

Moreover, we investigate the problem of \emph{countable satisfiability}, which asks for a countable team satisfying a formula, as well as \emph{finite satisfiability}, for which the team is induced as the set of traces of a finite Kripke structure.
(Note that there are countable teams that not finitely generated, and finitely generated teams that are not countable.)
For obvious reasons, the finite satisfiability problem is closer to practical applications than general satisfiability, and while these satisfiability notions coincide for classical LTL~\cite{SC85}, they do not for, \eg, HyperLTL~\cite{abs-1907-05070}.

We sum up our complexity theoretic results for $\LTL(\sneg)$.
The bold-face symbol $\mathbf{\Delta^3_0}$ refers to the decision problem of third-order arithmetic $\Delta^3_0$ (and likewise $\mathbf{\Delta^2_0}$ for second-order).
These and other notions will be defined in \Cref{sec:upper}.

\begin{theorem}\label{thm:uncountable-main-thm}
  Model checking, satisfiability and finite satisfiability are equivalent to $\mathbf{\Delta^3_0}$ and countable satisfiability is equivalent to $\mathbf{\Delta^2_0}$ \wrt logspace-reductions.
  Furthermore, this already holds with only the temporal operator $\F$ and temporal depth two (for satisfiability) and one (for model checking), respectively.
\end{theorem}

We also consider so-called $\calC$-restricted variants of the above problems, where $\calC$ is, \eg, the class of ultimately periodic traces.
Here, of all traces generated by a given structure or contained in a satisfying team, we consider only those in $\calC$.

\begin{theorem}\label{thm:countable-main-thm}
  If $\calC$ is the class of ultimately periodic traces or the class of ultimately constant traces, then $\calC$-restricted model checking, satisfiability and finite satisfiability are equivalent to $\mathbf{\Delta^2_0}$ \wrt logspace-reductions.
Again, this already holds with only the temporal operator $\F$ and temporal depth two (for satisfiability) and one (for model checking), respectively.
\end{theorem}

In \Cref{sec:upper}, we prove the upper bound of \Cref{thm:uncountable-main-thm}, which mostly amounts to a translation from $\LTL(\sneg)$ to $\Delta^3_0$ on the level of formulas, which reduces both model checking and satisfiability to $\mathbf{\Delta^3_0}$.
In \Cref{sec:mc}, we then present a reduction of $\mathbf{\Delta^3_0}$ back to model checking.
More precisely, given a $\Delta^3_0$-formula $\phi$, we compute a model $\calK$ and a temporal formula $\psi$ such that $\psi$ holds in the traces of $\calK$ iff $\phi$ is true in $\N$.
The reduction is surprisingly simple in the sense that it works already for a small fragment of formulas, namely with only $\F$-operators and no temporal nesting.

Next, in \Cref{sec:sat}, we reduce the problem of model checking to satisfiability.
It is a standard approach to reduce LTL model checking to (un-)satisfiability, and amounts to describing the structure of a model as a formula~\cite{Sch02} (or alternatively as an $\omega$-automaton~\cite{VW86}).
For $\LTL(\sneg)$, we do this only indirectly.
In fact, we show that the full team of all possible traces is definable, and from it extract the subteam of precisely the traces of the structure.
Like for model checking, already a weak fragment suffices for the hardness of satisfiability, namely only $\F$ and temporal depth two.

By the above chain of reductions, then all these problems are logspace-equivalent.
If the restriction is imposed that the teams in question are  countable, for example by looking only at ultimately periodic traces, then all the above arguments can be carried out with second-order arithmetic $\Delta^2_0$, by which we will prove \Cref{thm:countable-main-thm}.

 \section{Preliminaries}

\paragraph{Basic notions}
The power set of a set $X$ is $\wp X$.
The set of non-negative integers is $\N = \{0, 1, 2, \ldots\}$, also denoted by $\omega$.
For $n \in \N$, we write short $[n]$ for $\{1,2,\ldots,n\}$.
The set of all infinite sequences over $X$ is $X^\omega$.
The set of all $n$-tuples over $X$ is $X^n$, and the set of all finite sequences is $X^*$.
We sometimes write an infinite sequence $x = x_0,x_1,\ldots$ as $(x_i)_{i \geq 0}$, and refer to the $i$-th element $x_i$ also as $x(i)$.

\smallskip

\paragraph{Computational complexity}
A \emph{(decision) problem} is a subset $A \subseteq \Sigma^*$ of words over some finite alphabet $\Sigma$, which is assumed as $\Sigma = \{0,1\}$ unless stated otherwise.
A \emph{logspace-reduction} from a problem $A$ to a problem $B$ is a function $f$ that is computable in space $\mathcal{O}(\log n)$ such that $x \in A \LR f(x) \in B$, for all $x \in \Sigma^*$.
$A$ and $B$ are \emph{logspace-equivalent} if they are mutually logspace-reducible.

All stated reductions are logspace-reductions.
For a detailed introduction to computation theory and complexity, we refer the reader to standard literature~\cite{Sipser12}.

\paragraph{Linear Temporal Logic}
Let $\AP = \{ p_1, p_2, \ldots \}$ be a countably infinite set of atomic propositions.
The set of all formulas of \emph{linear temporal logic (LTL)} over $\AP$ is written $\LTL$, and is defined by the grammar
\begin{align*}
 \phi \ddfn p_i \mid \neg \phi \mid \phi \land \phi \mid \phi \lor \phi \mid \X \phi \mid \F \phi \mid \G \phi \mid \phi \U \phi \mid \phi \R \phi\text{.}
\end{align*}
The connectives $\X$ (\emph{nexttime}), $\F$ (\emph{future}), $\G$ (\emph{globally}), $\U$ (\emph{until}), $\R$ (\emph{release}) are called \emph{temporal operators}.
The \emph{temporal depth} $\td(\phi)$ is their nesting depth, that is,
\begin{alignat*}{3}
  & \td(p_i)  &  & \dfn 0  &  &  \\
  & \td(\neg\phi) &  & \dfn \td(\phi)  &  &  \\
  & \td(\phi \circ \psi) &  & \dfn \max\{\td(\phi),\td(\psi)\} &  & \text{ for }\circ \in \{\land,\lor\} \\
  & \td(O\phi)  &  & \dfn \td(\phi)+1  &  & \text{ for }O\in \{\X,\F,\G\}  \\
  & \td(\phi O \psi) &  & \dfn \max\{\td(\phi),\td(\psi)\} + 1\; &  & \text{ for } O\in \{\U,\R\}\text{.}
\end{alignat*}
The logic $\LTL_k(O_1,\ldots,O_n)$, for $n,k \geq 1$, is the syntactical fragment of $\LTL$ that contains all formulas of temporal depth up to $k$, and in which only the temporal operators $O_1,\ldots,O_n \in\{\X,\F,\G,\U,\R\}$ are used.

\paragraph{Traces}
A \emph{label} is a subset of $\AP$.
If $s$ is a label, then we also write $s^n$ or $s^\omega$ for length $n$ resp.\ infinite sequences consisting only of the label $s$.
A \emph{trace} is an element of $(\wp\AP)^\omega$, \ie, an infinite sequence of labels.
The suffix $t(i)t(i+1)\cdots$ is $t^i$, where $t^0 = t$.
The \emph{projection} of $t$ onto a set $\Phi \subseteq \AP$ is $\restr{t}{\Phi}$, defined by $\restr{t}{\Phi} \dfn (t(i)\cap \Phi)_{i \geq 0}$.

A trace is \emph{ultimately periodic} if there exist $c,d > 0$ such that $t(n) = t(n+d)$ for all $n \geq c$.
It is \emph{constant} if $t(0) = t(1) = \cdots$, and \emph{ultimately constant} if $t^i$ is constant for some $i$.

\paragraph{Semantics}
LTL-formulas are evaluated on traces as follows, where $t$ is a trace and $p \in \AP$.
\begin{alignat*}{5}
  & t \vDash p  & \;\LR \; &  p \in t(0) &  &  &  & t \vDash \G \phi & \;\LR \;  &  \forall k \geq 0 : t^k \vDash \phi  \\
  & t \vDash \neg \phi  & \; \LR \; &  t \nvDash \phi &  &  &  & t \vDash \phi \U \psi & \LR \; & \exists k \geq 0 : t^k \vDash \psi  \\
  & t \vDash \phi \land \psi & \;\LR \; &  t \vDash \phi \text{ and }t \vDash \psi &  & \qquad &  &  &  & \quad \text{ and } \forall j < k : t^j \vDash \phi \\
  & t \vDash \phi \lor \psi  & \;\LR \; &  t \vDash \phi \text{ or }t \vDash \psi  &  &  &  & t \vDash \phi \R \psi & \;\LR \; &  \forall k \geq 0 : t^k \vDash \psi  \\
  & t \vDash \X \phi  & \;\LR \; &  t^1 \vDash \phi  &  &  &  &  & & \quad \text{ or } \exists j < k : t^j \vDash \phi \\
  & t \vDash \F \phi  &\; \LR \; &  \exists k \geq 0 : t^k \vDash \phi &  &  &  &  &  &
\end{alignat*}

We employ the usual abbreviations, implication $\phi \imp \psi \equiv \neg \phi \lor \psi$, equivalence $\phi \equi \psi \equiv (\phi \imp \psi) \land (\psi \imp \phi)$, truth $\top \equiv p \lor \neg p$, and falsity $\bot \equiv p \land \neg p$.

As usual, we use $\vDash$ and $\equiv$ for semantical entailment and equivalence.

\paragraph{Kripke structures}
A \emph{Kripke structure} or just \emph{structure} is a tuple $\calK = (W, R, \eta, r)$ where $W$ is a non-empty set of \emph{states} or \emph{worlds}, $R \subseteq W \times W$ is the \emph{transition relation}, $\eta \colon W \to \wp \AP$ maps each state to a label, and $r \in W$ is the \emph{initial state} or \emph{root} of $\calK$.
We assume that all structures are \emph{serial}, or \emph{total}, that is, for every $w \in W$ there exists some $v \in W$ such that $(w,v) \in R$.
A \emph{path} in $\calK$ is a sequence $\pi \in W^\omega$ such that $(\pi(i),\pi(i+1)) \in R$ for all $i \geq 0$, and $\pi(0) = r$.

The \emph{trace induced by a path $\pi$} is $\eta(\pi) \dfn (\eta(\pi(i)))_{i \geq 0}$.
Note that paths and traces are different objects, and an aperiodic path may still induce a constant trace if it cycles through states with equal labels.

The set of all traces induced by paths in $\calK$ is $T(\calK)$.
A structure $\calK$ \emph{satisfies} a formula $\phi$, in symbols $\calK \vDash \phi$, if $t \vDash \phi$ for all $t \in T(\calK)$.
The classical model checking problem of LTL is now the set of all pairs $(\calK,\phi)$ where $\calK$ is a structure, $\phi$ is an LTL-formula, and $\calK \vDash \phi$.

The subset of ultimately constant, resp.\ periodic, traces in $\calK$ is denoted by $T_\ulc(\calK)$, resp.\ $T_\ulp(\calK)$.

\subsection{Team semantics}

A \emph{team} $T$ is a (possibly empty) set of traces, formally $T \subseteq (\wp \AP)^\omega$.
As for traces, we define the suffix $T^i \dfn \{ t^i \mid t \in T\}$, and say that a team $T$ is \emph{constant}, \emph{ultimately constant}, or \emph{ultimately periodic}, if all traces $t \in T$ are.
Also, we define the projection $\restr{T}{\Phi} \dfn \{ \restr{t}{\Phi} \mid t \in T \}$.

Next, we introduce \emph{synchronous team semantics} of $\LTL$ as defined by Krebs et al.~\cite{KMV018}.
Let $T$ be a team and $p \in \AP$.
Then:
\begin{alignat*}{2}
  & T \vDash p &  & \LR \; \forall t \in T : p \in t(0) \\
& T \vDash \neg \phi &  & \LR \; \forall t \in T : \{t\} \nvDash \phi \\
  & T \vDash \phi \land \psi &  & \LR \; T \vDash \phi \text{ and }T \vDash \psi  \\
  & T \vDash \phi \lor \psi  &  & \LR \; \exists S, U \subseteq T : S \cup U = T, S \vDash \phi \text{ and }U \vDash \psi \\
& T \vDash \X \phi &  & \LR \; T^1 \vDash \phi  \\
  & T \vDash \F \phi &  & \LR \; \exists k \geq 0 : T^k \vDash \phi \\
  & T \vDash \G \phi &  & \LR \; \forall k \geq 0 : T^k \vDash \phi \\
  & T \vDash \phi \U \psi &  & \LR \; \exists k \geq 0 : T^k \vDash \psi \text{ and } \forall j < k : T^j \vDash \phi \\
  & T \vDash \phi \R \psi &  & \LR \; \forall k \geq 0 : T^k \vDash \psi \text{ or } \exists j < k : T^j \vDash \phi
\end{alignat*}

Note that the disjunction of team semantics is not Boolean, but instead \emph{splits} the team into two parts.
For example, the team $\{ \{p\}\emptyset^\omega, \emptyset\{p\}\emptyset^\omega\}$ does satisfy $\F p \lor \F p$, as each trace itself satisfies $\F p$, but the whole team does not satisfy $\F p$, as there is no single timestep $k$ at which $T^k \vDash p$~\cite{KMV018}.
Hence $\F p \lor \F p \not\equiv \F p$.

On single traces, classical semantics and team semantics coincide:
\begin{proposition}[\cite{KMV018}]
  For all $\phi \in \LTL$ and traces $t$, $t \vDash \phi$ iff $\{t\}\vDash \phi$.
\end{proposition}

The following are standard properties of logics with team semantics:
\begin{definition}[Downward closure]
  A formula $\phi$ is \emph{downward closed} if for all teams $T$ and $T' \subseteq T$ it holds that $T \vDash \phi$ implies $T'\vDash \phi$.
\end{definition}

\begin{definition}[Empty team satisfaction]
  A formula $\phi$ has \emph{empty team satisfaction} if $\emptyset \vDash \phi$.
\end{definition}

\begin{definition}[Union closure]
  A formula $\phi$ is \emph{union closed} if for all sets $\calT$ of teams it holds that $\forall T \in \calT : T \vDash \phi$ implies $\bigcup \calT \vDash \phi$.
\end{definition}

\begin{definition}[Flatness]
  A formula $\phi$ is \emph{flat} if for all teams $T$ it holds that $T \vDash \phi \LR \forall t \in T : \{ t \} \vDash \phi$.
\end{definition}

We say that a logic $L$ has one of the above properties if all $L$-formulas have the respective property.

Observe that union closure implies empty team satisfaction, and that flatness is equivalent to simultaneous downward closure and union closure.

\begin{proposition}[\cite{KMV018}]
  $\LTL$ with team semantics has downward closure and empty team satisfaction.
\end{proposition}
This property is shared with other related logics such as \emph{dependence logic}~\cite{Vaa07}.

However, unlike asynchronous semantics, the synchronous semantics we use is not union closed and hence not flat:

\begin{proposition}[\cite{KMV018}]
  The formula $\F p$ is not union closed.
\end{proposition}

In fact, this already follows from the previous example, as in union closed logics $\phi \lor \phi \equiv \phi$, but here $\F p \lor \F p \nvDash \F p$.

Next, we introduce the \emph{dependence atom}\footnote{The term \emph{atom} is used for historic reasons.
In the first-order setting, the atom ranges only over individual variables and/or terms as arguments, and for this reason is indeed an atomic formula~\cite{Vaa07}.
That the arguments themselves are formulas happens only in proposition-based logics.
In the latter setting, usually, non-classical atoms such as the dependence atom range only over classical formulas, whereas here they range over arbitrary formulas. However, this does not affect our results, since \Cref{prop:dep-is-definable} and the other translations~\cite{LV19} also hold for this more general syntax.}
as an example for a non-classical formula in team-semantics.
For LTL, it is defined as follows~\cite{KMV018}:
\begin{align*}
  T \; \vDash \; &\;\dep{\phi_1,\ldots,\phi_n;\psi} \LR \forall t, t' \in T :\\
  & \qquad \qquad  (\forall i \in [n] : \{t\} \vDash \phi_i \LR \{t'\}\vDash \phi_i) \Rightarrow (\{t\} \vDash \psi \LR \{t'\} \vDash \psi)
\end{align*}
That is, whenever traces in $T$ agree on the truth of all $\phi_i$, then they agree on the truth of $\psi$, or in other words, $\psi$ is functionally determined by $\phi_1,\ldots,\phi_n$.
For the case $n = 0$ we just write $\dep{\psi}$, which means that the truth of $\psi$ is constant among the team.

Let $\LTL(\depopp)$ denote the extension of $\LTL$ by the dependence atom.
The following result is analogous to first-order dependence logic~\cite{Vaa07}.
\begin{proposition}
  $\LTL(\depopp)$ has downward closure and empty team satisfaction.
\end{proposition}

\begin{theorem}[\cite{KMV018}]
  Model checking of $\LTL(\depopp)$ is hard for $\NEXPTIME$.
\end{theorem}

\subsection{Team logic with negation}

In team logic, the contradictory negation, or Boolean negation, is denoted by $\sneg$ to distinguish it from $\neg$:
\begin{align*}
  T \vDash \sneg \phi \; \LR \; T \nvDash \phi
\end{align*}
We write $\LTL(\sneg)$ for the extension of $\LTL$ by $\sneg$, and define fragments of the form $\LTL_k(\sneg,O_1,\ldots,O_n)$ like for classical $\LTL$.

\smallskip

With the Boolean negation available, it is now possible to define the Boolean disjunction, the material implication, and equivalence on the level of teams:
\begin{align*}
  \phi \ovee \psi &\dfn \sneg(\sneg \phi \land \sneg \psi)\text{,}\\
  \phi \timp \psi &\dfn \sneg \phi \ovee \psi\text{,}\\
  \phi \tequiv \psi &\dfn (\phi \timp \psi) \land (\psi \timp \phi)\text{.}
\end{align*}

\label{p:team-connectives}

Observe that the formulas $\phi$ and $\neg \neg \phi$ are not equivalent, but $\neg \phi$ and $\neg \neg \neg \phi$ are (similarly to intuitionistic logic).
The reason is that $\neg \neg \phi$ states that $\phi$ holds in all singletons.\footnote{In the team logic literature, usually $\neg$ is allowed only in front of atoms $p_i$. In such cases, $\neg \phi$ is sometimes defined as an abbreviation for pushing $\neg$ inwards to the atomic level using the equivalences $\neg (\phi\land\psi) \equiv \neg \phi \lor \neg \psi$,
$\neg \F \phi \equiv \G \neg \phi$, and so on.
However, it will be useful to define $\neg$ in front of arbitrary formulas in a way that is consistent with its semantics on classical formulas, for example in order to succinctly define a "flat" approximation of a formula.
With our definition, we follow Yang et al.~\cite{Yang17,VY17,Iemhoff016} and Kuusisto~\cite{Kuusisto15}.
}

Furthermore, non-classical atoms such as the dependence atom become definable~\cite{LV19}.
\begin{proposition}\label{prop:dep-is-definable}
  The dependence atom is definable as follows:
\emph{\begin{align*}
  \dep{\phi_1,\ldots,\phi_n;\psi} &\equiv \sneg \bigg( \top \lor \Big(\bigwedge_{i=1}^n\dep{\phi_i} \land \sneg \dep{\psi} \Big) \bigg)
  \intertext{\emph{where}}
  \dep{\phi} &\equiv \neg \neg \phi \ovee \neg \phi\text{.}
\end{align*}}
\end{proposition}
\begin{proof}
  First, $\dep{\phi}$ says that every trace satisfies either $\phi$ or $\neg \phi$, and hence that $\phi$ is constant among all traces.
  For this reason, the part inside the outermost $\sneg$ of the first line states that some subteam is constant in every $\phi_i$, but not in $\psi$, which is precisely the negation of the dependence atom.
\end{proof}

\begin{corollary}
  $\LTL(\sneg)$ is neither downward closed nor union closed.
  Also, this even holds without any temporal operators.
\end{corollary}
\begin{proof}
  The formula $\dep{p}$ is not union closed, and $\sneg p$, which states that at least one trace in the team satisfies $\neg p$, is not downward closed.
\end{proof}

In fact, also the connective $\neg$ can be expressed with $\sneg$, $\bot$ and $\lor$ as follows~\cite{Luck20}.
First note that the constant $\top$ is equivalent to $\sneg(p \land \sneg p)$.
Moreover, $\bot$ holds only in the empty team, so non-emptiness can be expressed as $\sneg \bot$.
Now the formulas\label{p:quantifiers}
\begin{alignat*}{3}
  &\qee{}\phi \dfn \top \lor \phi && \qaa{}\phi \dfn \sneg \qee{} \sneg \phi
  \intertext{say that some (resp.\ every) subteam satisfies $\phi$, and}
  &\qe{}\phi \dfn \qee{}(\sneg \bot \land \qaa{}(\bot \ovee \phi)) & \qquad\qquad & \qa{}\phi \dfn \sneg \qe{} \sneg \phi
\end{alignat*}
say that some (resp.\ every) singleton subteam satisfies $\phi$.
Intuitively, the formula $\qe{}\phi$ states that there is some non-empty subteam of which all non-empty subteams satisfy $\phi$, which is precisely the case if some singleton satisfies $\phi$.
Thus we can write $\neg \phi$ as $\qa{}\sneg\phi$.\label{p:neg-to-sneg}

Sometimes it is necessary to "condition" a team $T$ to only those traces that satisfy a certain formula $\phi$.
For this, we define the team $T_{\phi} \dfn \{ t \in T \mid \{t\} \vDash \phi \}$.
Observe that $T_{\phi}$ and $T_{\neg \phi}$ always form a disjoint partition of $T$.

On the formula side, we use the connective $\phi \hook \psi$ to express that $T_{\phi}\vDash \psi$.\label{p:hook}
It is definable as $\neg\phi \lor (\neg\neg \phi \land \psi)$.

\medskip

Let us turn to the decision problems associated with $\LTL(\sneg)$.
A formula is \emph{satisfiable} if it is true in at least one team.
The \emph{satisfiability problem} of a logic $L$ formally is the set of all satisfiable formulas $\phi \in L$.

A team $T$ is \emph{finitely generated} if $T = T(\calK)$ for some finite structure $\calK$, and a formula $\phi$ is \emph{finitely satisfiable} if it satisfied by some finitely generated team.
Also, $\phi$ is called \emph{countably satisfiable} if it is satisfied by some countable team.

For practical purposes, it is sometimes preferable to consider only traces that have a finite representation, which is the case for example for ultimately periodic traces.
In general, if $\calC$ is some class of traces, then we say that a formula $\phi$ is \emph{$\calC$-satisfiable} if there is a team $T$ such that $T \cap \calC$ satisfies $\phi$.
Likewise, $\phi$ is \emph{finitely $\calC$-satisfiable} if $T(\calK) \cap \calC$ satisfies $\phi$ for some finite structure $\calK$.

Note that it makes no sense to combine the restriction of finite satisfiability and countability alone:

\begin{proposition}\label{prop:fin-count-is-ulp}
If $\calK$ is a finite structure and $T(\calK)$ is countable, then $T(\calK)$ is already ultimately periodic.
\end{proposition}
\begin{proof}
  Proof by contraposition:
  Suppose some $t \in T(\calK)$ is not ultimately periodic.
  The path $\pi$ that induces $t$ must visit some state $w$ in $\calK$ infinitely often, say at positions $i_0,i_1,\ldots$.
  If we now consider the intervals $t(i_0)\cdots t(i_1-1)$, $t(i_1)\cdots t(i_2-1), \ldots$, then at least two of them are distinct, since $t$ is not ultimately periodic.
  But then there are at least two cycles through $w$ with distinct labels, and hence $T(\calK)$ is uncountable.
\end{proof}

\smallskip

The \emph{model checking problem} of $L$ is the set of all pairs $(\calK,\phi)$ such that $\phi \in L$ and $T(\calK) \vDash \phi$ in team semantics.\footnote{Here, formulas are written as strings over a suitable finite alphabet, with propositions encoded as binary numbers.
A finite Kripke structure $\calK = (W,R,\eta,r)$ is encoded by a number $n \dfn \size{W}$, lists of pairs of numbers (in binary) for with $R$ and $\eta$, and a single number for $r$.
}

Again, for a class $\calC$ of traces, the problem of $\calC$-model checking is the set of pairs $(\calK,\phi)$ where $T(\calK) \cap \calC \vDash \phi$.
Recall that $T_\ulp(\calK)$ is the subteam of all ultimately periodic traces in $T(\calK)$ (and analogously $T_\ulc(\calK)$), and can thus be written as $T(\calK)\cap \calC$ for appropriate $\calC$.\footnote{Classically, there is no difference between full and ultimately periodic model checking.
The reason is that every satisfiable $\LTL$-formula is satisfied in some ultimately periodic trace~\cite{SC85}, and that the traces of a structure $\calK$ are definable by a formula $\chi_\calK$~\cite{Sch02}.
So if some trace $t$ in $\calK$ satisfies $\neg \phi$, then $\chi_\calK \land \neg \phi$, and hence $\neg \phi$, holds in some ultimately periodic trace in $\calK$.}
For a class $\calC$ of traces, the corresponding decision problems are called \emph{$\calC$-restricted} model checking, satisfiability, etc.
If $\calC$ is the class of ultimately periodic (constant) traces, then we just call the $\calC$-restricted model checking problem \emph{ultimately periodic} and \emph{ultimately constant} model checking, respectively.

\smallskip

For flat formulas, the complexity of these problems coincides with their classical counterparts.
The underlying reductions for their lower bounds are all computable in logspace.

\begin{proposition}[\cite{SC85}]
  Model checking and satisfiability of classical $\LTL$ is $\PSPACE$-complete.
\end{proposition}

\begin{proposition}[\cite{KMV018}]
The satisfiability problem for $\LTL$-formulas in team semantics, restricted to non-empty teams, is $\PSPACE$-complete.
\end{proposition}
Here, the empty team needs to be excluded since otherwise the problem becomes trivial due to empty team satisfaction.
With negation $\sneg$, this distinction becomes redundant, since $\phi$ is satisfiable iff $\phi \lor \top$ is satisfiable in a non-empty team, and $\phi$ is satisfiable in a non-empty team iff $\phi \land \sneg\bot$ is satisfiable.

\begin{proposition}[\cite{KMV018}]
  The model checking problem for $\LTL$-formulas in team semantics is $\PSPACE$-hard.
\end{proposition}
 
\section{A Stutter-Invariant Fragment of $\LTL(\sneg)$}\label{sec:stutter}

We begin by investigating the concept of \emph{stutter-equivalence} and lift the classical definition to team semantics.
We follow Peled and Wilke~\cite{PW97}.
Two traces $t,t'$ are called \emph{stutter-equivalent} if there are two sequences $0 = i_0 < i_1 < i_2 < \cdots$ and $0 = j_0 < j_1 < j_2 < \cdots$ of indices such that, for all $k \geq 0$, it holds that $t(i_k) = \cdots = t(i_{k+1}-1) = t'(j_k) = \cdots t'(j_{k+1}-1)$.

Intuitively, two traces are stutter-equivalent if one can be obtained from the other by adding and removing consecutive copies of labels (while leaving at least one copy), \eg, $\{p\}\{p\}\{q\}\{z\}^\omega$ is stutter-equivalent to $\{p\}\{q\}\{q\}\{z\}^\omega$, but $\{p\}\{z\}^\omega$ is not.

Every $\X$-free $\LTL$-formula $\phi$ defines a \emph{stutter-invariant} property in the sense that it cannot distinguish stutter-equivalent traces $t$ and $t'$, \ie, $t \vDash \phi \LR t' \vDash \phi$.
Indeed, Peled and Wilke~\cite{PW97} show that the stutter-invariant properties definable in $\LTL$ are exactly those definable without $\X$.

\subsection{Stutter-equivalence in team semantics}

The first step in this section is to generalize the definition of stutter-equivalence to teams.
For this, note that the above definition of stutter-equivalence can also be written as follows:
For all $k$, it holds that $t(i_k) = \cdots = t(i_{k+1}-1)$ and $t'(j_k) = \cdots = t'(j_{k+1}-1)$, and additionally the "filtered" traces $(t_{i_k})_{k\geq 0} = t(i_0)t(i_1)t(i_2)\cdots$ and $(t'_{j_k})_{k \geq 0} =  t'(j_0)t'(j_1)t'(j_2)\cdots$ are identical.
We generalize this equivalent definition to teams:

\begin{definition}[Stuttering functions]
 A \emph{stuttering function of a trace $t$} is a strictly increasing function $f \colon \N \to \N$ such that $f(0) = 0$ and $t(f(k)) = \cdots = t(f(k+1)-1)$ for all $k$.
 A \emph{stuttering function of a team $T$} is a function $f$ that is a stuttering function for each trace $t \in T$.
\end{definition}

In other words, the range of a stuttering function includes at least zero and all positions on which the trace differs from the predecessor.

\paragraph{Example}
Every trace of the form $\emptyset^n\{p\}\emptyset^\omega$ has infinitely many stuttering functions, whose range only needs to include zero, $n+1$, and an arbitrary infinite subset of $\{n+2,\ldots\}$.
On the other hand, when combined to a team, the only stuttering function of $\{ \emptyset^n\{p\}\emptyset^\omega \mid n \geq 0 \}$ is the identity;
intuitively, there exists no position $i$ on which \emph{every} trace stays unchanged compared to the previous position.

If $f \colon \N \to \N$, then let $t[f]$ denote the trace $t(f(0))t(f(1))t(f(2))\cdots$, and $T[f]$ the team $\{ t[f] \mid t \in T\}$.

\begin{definition}[Stutter-equivalence]
 Teams $T,T'$ are \emph{stutter-equivalent}, in symbols $T \steq T'$, if there are stuttering functions $f$ of $T$ and $f'$ of $T'$ such that $T[f] = T'[f']$.
\end{definition}
See also \Cref{fig:def-stutter-equiv} for two example teams that are stutter-equivalent.

\begin{figure}
 \centering
 \begin{tikzpicture}[node distance=1,inner sep=1mm,scale=1,thick,->,>=stealth]
  \tikzset{w/.style={draw,circle,inner sep=.8mm,black,fill}}
  \tikzset{a/.style={fill=black!20!white}}
  \tikzset{b/.style={fill=black!90!white}}
  \tikzset{fun/.style={node distance=1.8mm}}
  \tikzset{sec/.style={draw,rounded corners,inner sep=2mm,dotted,fill=black,fill opacity=.07}}

  \node[w,b] (t0) at (0,-2) {};
  \foreach \i\p in {1/b, 2/a, 3/a, 4/b, 5/b, 6/a}{
    \pgfmathtruncatemacro{\pre}{\i - 1}
    \node[w,\p,right=of t\pre] (t\i) {};
    \draw[-] (t\pre) to (t\i);
  }
  \node[right=of t6] (t7) {$\cdots$};
  \draw[->] (t6) to (t7);

  \node[w,a] (s0) at (0,-1){};
  \foreach \i\p in {1/a, 2/a, 3/a, 4/a, 5/a, 6/b}{
    \pgfmathtruncatemacro{\pre}{\i - 1}
    \node[w,\p,right=of s\pre] (s\i) {};
    \draw[-] (s\pre) to (s\i);
  }
  \node[right=of s6] (s7) {$\cdots$};
  \draw[->] (s6) to (s7);

  \node[w,b] (u0) at (0,1){};
  \foreach \i\p in {1/a, 2/a, 3/b, 4/a, 5/a, 6/a}{
    \pgfmathtruncatemacro{\pre}{\i - 1}
    \node[w,\p,right=of u\pre] (u\i) {};
    \draw[-] (u\pre) to (u\i);
  }
  \node[right=of u6] (u7) {$\cdots$};
  \draw[->] (u6) to (u7);

  \node[w,a] (v0) at (0,2){};
  \foreach \i\p in {1/a, 2/a, 3/a, 4/b, 5/b, 6/b}{
    \pgfmathtruncatemacro{\pre}{\i - 1}
    \node[w,\p,right=of v\pre] (v\i) {};
    \draw[-] (v\pre) to (v\i);
  }
  \node[right=of v6] (v7) {$\cdots$};
  \draw[->] (v6) to (v7);

  {\scriptsize
  \node[fun,above=of s0] (n0) {$0\mapsto 0$};
  \node[fun,above=of s2] {$1 \mapsto 2$};
  \node[fun,above=of s4] {$2 \mapsto 4$};
  \node[fun,above=of s6] {$3 \mapsto 6$};
  \node[fun,below=of u0] (n1) {$0 \mapsto 0$};
  \node[fun,below=of u1] {$1 \mapsto 1$};
  \node[fun,below=of u3] {$2 \mapsto 3$};
  \node[fun,below=of u4] {$3 \mapsto 4$};
  \node[left=-2mm of n1] {$f\colon$};
  \node[left=-2mm of n0] {$g\colon$};
  }

  \node[sec,fit=(u0)(v0)] {};
  \node[sec,fit=(t0)(s1)] {};
  \node[sec,fit=(u1)(v2)] {};
  \node[sec,fit=(s2)(t3)] {};
  \node[sec,fit=(u3)(v3)] {};
  \node[sec,fit=(t4)(s5)] {};
  \node[sec,fit=(t6)(s6)] {};
  \node[sec,fit=(u4)(v6)] {};

  \node at (-1,1.5) {\Large$T$};
  \node at (-.95,-1.5) {\Large$T'$};

 \end{tikzpicture}
 \caption{Stutter-equivalence of teams witnessed by stuttering functions $f$ and $g$.
 Distinct propositional labels are encoded by different colors.\label{fig:def-stutter-equiv}}
\end{figure}
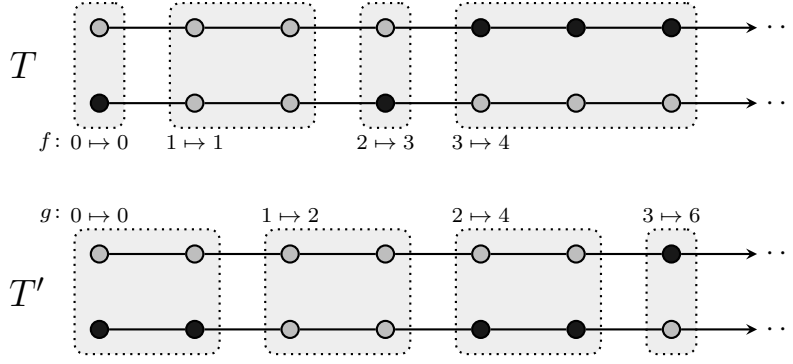

Next, we show that stutter-equivalent teams always have the same cardinality.

\begin{proposition}
 If $f$ is a stuttering function of $T$, then $\size{T} = \size{T[f]}$.
\end{proposition}
\begin{proof}
We show that the surjective map $t \mapsto t[f]$ from $T$ to $T[f]$ is also injective.
Let $t,t' \in T$ such that $t[f] = t'[f]$.
 Then $t(f(k)) = t'(f(k))$ for all $k$.
 As $f$ is a stuttering function, also $t(j) = t(f(k)) = t'(f(k)) = t'(j)$ for all $j$ such that $f(k) \leq j < f(k+1)$ for some $k$, and hence for all $j$.
 Consequently, $t = t'$.
\end{proof}

\begin{corollary}
 Any two stutter-equivalent teams have the same cardinality.
\end{corollary}

In our definition, a team is a set of traces, which consist of labels indexed by natural numbers.
However, a team $T$ also intuitively corresponds to a \emph{single} sequence which maps to each position a "snapshot" of all traces, indexed by the elements of $T$.
Let us relax the necessity of $\AP$ being countable in this section, then the resulting sequence is indeed a trace.

Fix an index set $I$.
The \emph{snapshot} trace $\mathsf{sn}(T)$ of an $I$-indexed team $\{t_i\}_{i \in I}$ is a trace in $(\wp(\AP \times I))^\omega$ and is defined by $(p,i) \in \mathsf{sn}(T)(n) \LR p \in t_i(n)$.

\begin{proposition}
 Let $T = \{t_i\}_{i \in I}$ and $T' = \{t'_i\}_{i \in I}$ be $I$-indexed teams.
 Then $T$ and $T'$ are stutter-equivalent teams if and only if $\mathsf{sn}(T)$ and $\mathsf{sn}(T')$ are stutter-equivalent traces.
\end{proposition}
\begin{proof}
 Clearly, a function $f$ is a stuttering function of a team $T$ if and only if it is one of the trace $\mathsf{sn}(T)$.
 Also, clearly $\mathsf{sn}(T)[f] = \mathsf{sn}(T[f])$.
 As a consequence, $T[f] = T'[g]$ for stuttering functions $f,g$ if and only if $\mathsf{sn}(T)[f] = \mathsf{sn}(T')[g]$.
\end{proof}

The above statement can be extended to all teams, since stutter-equivalent teams have the same cardinality and thus can be indexed by the same set $I$.

A team $T$ has a \emph{stuttering position $i$} if $t(i) = t(i+1)$ for all $t \in T$, but there are $t \in T$ and $j > i+1$ such that $t(i) \neq t(j)$.
(So constant suffixes are not considered stuttering.)
A trace or team with no stuttering positions is called \emph{stutter-free}.

Next, we prove the team analogs of several classical properties of stutter-equivalence.

\begin{theorem}\label{thm:stutter-properties}
 Stutter-equivalence on teams satisfies the following properties.
 \begin{enumerate}
  \item If $T$ is a stutter-free team, then $T = T[f]$ for every stuttering function $f$.
  \item Every team $T$ has a unique stutter-free team that is stutter-equivalent to it.
  \item Stutter-equivalence is reflexive, transitive and symmetric.
  \item Every stutter-equivalence class contains exactly one stutter-free team.
 \end{enumerate}
\end{theorem}
\begin{proof}
 1.: Suppose $T \neq T[f]$.
 Then there exists a minimal $i$ such that $f(i) \neq i$.
 In particular, $f(i) > i > 0$.
 Let $t \in T$ be arbitrary.
 By minimality of $i$, $t(i-1) = t(f(i-1))$.
 By definition of stuttering function, $t(f(i-1)) = \cdots = t(f(i)-1)$.
 Since $f(i-1) \leq i \leq f(i)-1$, it follows that $t(i-1) = t(i)$.
 But $T$ was stutter-free, contradiction.

 2.: Let $j_0,j_1,\ldots$ be all non-stuttering positions of $T$  in ascending order.
 Then $f(0) \dfn 0, f(i+1) \dfn j_i+1$ is a stuttering function of $T$ (see \Cref{fig:def-stuttering-function}).
 So $T \steq T[f]$ via the stuttering function $f$ for $T$ and the identity for $T[f]$.
 Moreover, the team $T[f]$ is stutter-free.
For the uniqueness part, suppose $T_1 \steq T$ and $T_2 \steq T$ for stutter-free $T_1 \neq T_2$ via stuttering functions $f_1$ for $T_1$ and $f_2$ for $T_2$.
 By 1., then $T[f_1] = T_1$ and $T[f_2] = T_2$.
 By distinctness, there is $t^\star \in T$ and a minimal $i > 0$ with $t^\star(f_1(i)) \neq t^\star(f_2(i))$, hence $f_1(i) \neq f_2(i)$.
 \Wloss $f_1(i) > f_2(i)$.
 By minimality of $i$, $f_1(i-1)=f_2(i-1)<f_2(i)\leq f_1(i)-1$, and by definition of stuttering function, $t(f_2(i-1))=t(f_2(i))$ for all $t \in T$.
 As $T[f_2]$ is stutter-free, the suffix $T[f_2]^{i-1}$ must already be constant.
 But then also the suffix $T^{f_2(i-1)}$ of $T$ is constant, contradiction to $t^\star(f_1(i)) \neq t^\star(f_2(i))$.

 3.: We show that two teams are stutter-equivalent if and only if they are stutter-equivalent to a common stutter-free team.
 This relation is transitive due to 2., and it is clearly reflexive and symmetric.
 The direction "$\Leftarrow$" follows from 1.
 For "$\Rightarrow$", suppose $T_1 \steq T_2$, and $T' \dfn T_1[f_1] = T_2[f_2]$ for stuttering functions $f_1,f_2$.
 By 1. and 2., there exist $g$ and a stutter-free team $T^\star$ with $T^\star \steq T'$ via $T[g] = T^\star$.
 Now, it is not hard to check that $f_i \circ g$ is a stuttering function of $T_i$, $i \in \{1,2\}$.
 As a result, $T_1\steq T^\star \steq T_2$.

 4.: Follows from 2. and 3.
\end{proof}

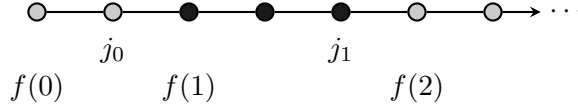
\begin{figure}
 \centering
 \begin{tikzpicture}[node distance=5mm,inner sep=1mm,scale=1,thick,->,>=stealth]
  \tikzset{w/.style={draw,circle,inner sep=.8mm,black,fill}}
  \tikzset{a/.style={fill=black!20!white}}
  \tikzset{b/.style={fill=black!90!white}}
  \node[w,a] (1) at (0,0){};
  \node[w,a] (2) at (1,0){};
  \node[w,b] (3) at (2,0){};
  \node[w,b] (4) at (3,0){};
  \node[w,b] (5) at (4,0){};
  \node[w,a] (6) at (5,0){};
  \node[w,a] (7) at (6,0){};
  \node (8) at (7,0){$\cdots$};

  \node[below= 1mm of 2]{$j_0$};
  \node[below= 1mm of 5]{$j_1$};

  \draw[->]
  (1) to (2)
  (2) to (3)
  (3) to (4)
  (4) to (5)
  (5) to (6)
  (6) to (7)
  (7) to (8)
  ;
  \node at (0,-1) {$f(0)$};
  \node at (2,-1) {$f(1)$};
  \node at (5,-1) {$f(2)$};
 \end{tikzpicture}
 \caption{Defining a stuttering function $f$ from the non-stuttering positions $j_i$.\label{fig:def-stuttering-function}}
\end{figure}

\subsection{The stutter-invariance of $\X$-free formulas}

In the remainder of this section, we prove that the $\X$-free formulas of $\LTL(\sneg)$ are indeed stutter-invariant, where a formula $\phi \in \LTL(\sneg)$ is \emph{stutter-invariant} if, for all teams $T,T'$, $T \steq T'$ implies $T \vDash \phi \LR T' \vDash \phi$.

The proof requires a series of technical lemmas.

\begin{lemma}[Splitting preserves stutter invariance]\label{lem:stutter-split}
 Let $T,S$ be teams.
 Then the following are equivalent:
 \begin{enumerate}
  \item $T \steq S$.
  \item For all $T_1,T_2$ such that $T = T_1 \cup T_2$, there are teams $S_1,S_2$ such that $S = S_1 \cup S_2$ and $T_i \steq S_i$ for $i \in \{1,2\}$.
 \end{enumerate}
\end{lemma}
\begin{proof}
 2.\ to 1.\ is easy:
 Obviously $T = T \cup \emptyset$, but $\emptyset$ is only stutter-equivalent to itself, so $T \steq S$.

 We proceed with 1.\ to 2.
 By assumption, there are stuttering functions $f$ of $T$ and $g$ of $S$ such that $T[f] = U = S[g]$.
 Let $U_i \dfn T_i[f]$ for $i \in \{1,2\}$, and let $S_i \dfn \{ t \in S \mid t[g] \in U_i \}$.
 Then $S_i[g] = U_i = T_i[f]$.
 Clearly $f$ and $g$ are stuttering functions for $T_i$ and $S_i$ as well, which yields $T_i \steq S_i$.

 Obviously $S_1 \cup S_2 \subseteq S$, so it remains to show that $S \subseteq S_1 \cup S_2$.
 Let $t \in S$.
 Then $t[g] \in S[g] = T[f]$, so there is $t' \in T$ such that $t'[f] = t[g]$.
 As $T \subseteq T_1 \cup T_2$, there is $i \in \{1,2\}$ such that $t' \in T_i$, consequently $t[g] = t'[f] \in U_i$.
 From $t \in S$ and $t[g] \in U_i$ we conclude $t \in S_i$.
\end{proof}

\begin{lemma}[Future preserves stutter invariance]\label{lem:mapping-stutter}
 Let $T,S$ be teams.
 Then the following are equivalent:
 \begin{enumerate}
  \item $T \steq S$.
  \item There are surjective, non-decreasing $\mu,\nu \colon \N \to \N$ such that $T^{\mu(n)} \steq S^{\nu(n)}$ for all $n \geq 0$.
 \end{enumerate}
\end{lemma}
In particular, 2.\ implies that every suffix of $T$ is stutter-equivalent to a suffix of $S$ and vice versa, which will be necessary for stutter invariance of the temporal operators.

\begin{proof}
 Again, 2.\ to 1.\ is easy, since necessarily $\mu(0) = \nu(0) = 0$.
For 1.\ to 2., suppose $T[f] = S[g]$ for stuttering functions $f,g$.
 We construct the functions $\mu,\nu$ inductively such that always $T^{\mu(n)} \steq S^{\nu(n)}$.
At the same time, we will ensure another invariant necessary for showing the correctness of the construction, namely that for every $n$ there is $k$ such that $f(k) \leq \mu(n) < f(k+1)$ and $g(k) \leq \nu(n) < g(k+1)$.

 Intuitively, the functions $\mu$ and $\nu$ "run along" $T$ and $S$, with $\mu$ "waiting" on $T$ if $S$ stutters and vice versa (recall that $\mu$ and $\nu$ are only non-decreasing, not necessarily increasing).
 Positions of the form $f(k)$ and $g(k)$, in the range of the stuttering functions, are then crossed in lockstep, so to speak.

 We give the formal definition.
 First, let $\mu(0) = \nu(0) = 0$.
 For the inductive step to $n+1$, we distinguish two cases.
First, if either $\mu(n) + 1 \notin \ran f$ or $\nu(n) + 1 \notin \ran g$ or both holds, let
 \[
  \mu(n+1) \dfn \begin{cases}
   \mu(n)   & \text{if }\mu(n)+1 \in \ran f    \\
   \mu(n)+1 & \text{if }\mu(n)+1 \notin \ran f
  \end{cases}
 \]
 and analogously for $\nu$ and $g$.
 In other words, we can "advance" $\mu$ and/or $\nu$ as long as neither crosses the range of $f$ or $g$.
 In either case, clearly $T^{\mu(n)}\steq T^{\mu(n+1)}$ and $S^{\nu(n)} \steq S^{\nu(n+1)}$, so $T^{\mu(n+1)} \steq S^{\nu(n+1)}$ by induction hypothesis.

 In the other case, both $\mu(n) + 1 \in \ran f$ and $\nu(n) + 1 \in \ran g$.
 Then we advance both $\mu$ and $\nu$ in lockstep:
 let $\mu(n+1) \dfn \mu(n)+1$ and $\nu(n+1) \dfn \nu(n)+1$.
 By induction hypothesis, there is a common $k$ such that $\mu(n+1) = \mu(n)+1 = f(k)$ and $\nu(n+1) = \nu(n)+1 = g(k)$.
 It remains to show that $T^{f(k)} \steq S^{g(k)}$.

 For this, we define stuttering functions $f_k$ of $T$ and $g_k$ of $S$ such that $T^{f(k)}[f_k] = S^{g(k)}[g_k]$ by $f_k(i) \dfn f(i + k) - f(k)$ and $g_k(i) \dfn g(i+k) - g(k)$.
 Clearly $f_k(0) = g_k(0) = 0$ and both are strictly increasing.

 We show that these are stuttering functions of $T^{f(k)}$ and $S^{g(k)}$.
 For the sake of contradiction, suppose $f_k$ is not a stuttering function of $T^{f(k)}$ (the proof for $g_k$ and $S^{g(k)}$ is analogous).
 Then there exist a position $0 < i \notin \ran f_k$ and a trace $t \in T^{f(k)}$ such that $t(i-1) \neq t(i)$, or equivalently, $t \in T$ such that $t(f(k)+i-1) \neq t(f(k)+i)$.

 However, this implies that $f(k) + i  \in \ran f$ as $f$ is a stuttering function of $T$, so $f(k) + i = f(\ell)$ for some $\ell$.
 In particular, $\ell > k$.
 But then $f_k(\ell - k) = f(\ell - k + k) - f(k) = f(\ell)-f(k) = i$, so $i \in \ran f_k$, contradiction.

 With $f_k$ and $g_k$ being stuttering functions, it remains to show $T^{f(k)}[f_k] = S^{g(k)}[g_k]$.
 For this we prove $T^{f(k)}[f_k] = T[f]^k$ and $S^{g(k)}[g_k] = S[g]^k$, since $T[f] = S[g]$ by assumption, which implies $T[f]^k = S[g]^k$.
 We show $T^{f(k)}[f_k] = T[f]^k$ (the proof is again analogous for $S$).
 For this it suffices that for every trace $t$,
 \begin{align*}
  t^{f(k)}[f_k] \,=\, (t(f_k(i)+f(k)))_{i \geq 0} \,=\, (t(f(i+k)))_{i \geq 0} \,=\, t[f]^k\text{.}\tag*{\qedhere}
 \end{align*}
\end{proof}

\begin{theorem}\label{thm:stutter-invariance-ltl}
 Every $\X$-free $\LTL(\sneg)$-formula is stutter-invariant.
\end{theorem}
\begin{proof}
 Due to the equivalences $\G\psi \equiv \sneg \F\sneg\psi$, $\F\psi \equiv \top\U\psi$ and $\psi\R\theta\equiv\sneg(\sneg\psi\U\sneg\theta)$, it suffices to consider formulas $\phi \in \LTL(\sneg,\U)$.
 We have to show, for teams $T$ and $S$, that $T \steq S$ implies $T \vDash \phi \LR S \vDash \phi$.
 Hence let $T$ and $S$ be stutter-equivalent teams via $T[f] = S[g]$.
The proof is now by induction on $\phi$.

 \begin{itemize}
  \item For all propositional formulas, \eg, $\phi = p, \neg p, \top, \bot$ for $p  \in \AP$, this is clear:
        Since $f(0) = 0 = g(0)$, the teams $T$ and $S$ agree on the first position of traces.
  \item For the Boolean connectives, $\land$ and $\sneg$, the induction step is clear.
  \item For the case $\phi = \neg \psi$, recall from p.~\pageref{p:neg-to-sneg} that $\neg \psi \equiv \qa{}\sneg\psi$, which can be reduced to the connectives $\bot$, $\sneg$ and $\lor$.
  \item The $\lor$-case follows by induction hypothesis and \Cref{lem:stutter-split}.
  \item For the $\U$-case, we show "$\Rightarrow$", as "$\Leftarrow$" is symmetric.
        Suppose $T \vDash \psi \U \theta$, so there is $k \geq 0$ such that $T^k \vDash \theta$ and $T^j\vDash \psi$ for all $j < k$.
        We show that $S \vDash \psi \U \theta$.
        By \Cref{lem:mapping-stutter}, there are surjective, non-decreasing $\mu,\nu$ such that $T^{\mu(n)}\steq S^{\nu(n)}$ for all $n$.
        In particular, $\mu(n) = k$ for some $n$, so by induction hypothesis $S^{\nu(n)} \vDash \theta$.
It remains to show that $S^\ell \vDash \psi$ for all $\ell < \nu(n)$.
        Choose $n$ minimal, \ie, such that $n = 0$ or $\mu(n-1) < k$.
        If $n = 0$, then $\nu(n) = 0$ and we are done.
        Otherwise let $\ell < \nu(n)$ be arbitrary; then there exists $m < n$ such that $\nu(m) = \ell$ since $\nu$ is surjective and non-decreasing.
        As $\mu(m) \leq \mu(n-1) < k$, by assumption $T^{\mu(m)} \vDash \psi$.
        But then also $S^{\nu(m)} = S^{\ell} \vDash \psi$ by induction hypothesis.
        \qedhere
 \end{itemize}
\end{proof}

\section{From LTL($\sneg$) to third-order arithmetic}\label{sec:upper}

In this section, we translate formulas of $\LTL(\sneg)$ into arithmetic formulas and by this obtain the upper complexity bounds for \Cref{thm:uncountable-main-thm} and \ref{thm:countable-main-thm}.
Before we start, let us briefly introduce third-order arithmetic.
For a full formal definition in the context of higher order logic, we refer the reader to the literature~\cite{Leivant94}; here we need only a small part of it.
In this context, we write $\N$ for the standard model of arithmetic with the usual interpretations for $+$, $\times$ and so on.

\smallskip

A (third-order) \emph{type} is a tuple $\tau = (n_1,\ldots,n_k)$, where $k,n_1,\ldots,n_k \geq 1$.
For each type $\tau$ there is a set of (third-order) \emph{$\tau$-variables} $\calV_\tau \dfn \{\fraka,\frakb,\frakc, \ldots \}$.
Syntactically, third-order logic extends second-order logic as follows.
We include all atomic formulas, connectives and quantifiers from second-order logic.
Moreover, if $\fraka \in \calV_\tau$ is a third-order variable of type $\tau = (n_1,\ldots,n_k)$ and for each $i \in \{1,\ldots,k\}$, $A_i$ is a second-order relation symbol of arity $n_i$, then $\fraka(A_1,\ldots,A_k)$ is an atomic formula.\footnote{We leave out higher-order functions here, as they easily can be represented by a relation encoding their graph, analogously to the second-order case.}
Finally, if $\phi$ is a formula and $\fraka$ is as above, then $\exists \fraka \phi$ and $\forall \fraka \phi$ are formulas.
A formula is \emph{closed} if it has no free variables.

If $\tau = (n_1,\ldots,n_k)$ is a type, then elements $\calA \in \wp(\wp\N^{n_1}\times \cdots \times \wp\N^{n_k})$ are called \emph{$\tau$-objects}.
For example, the subset relation $\subseteq$ is a $(1,1)$-object, and "${\subseteq}(A,B)$" would be an atomic formula if $A$ and $B$ are unary (set) variables.
An example of a $(2,2)$-object is the relation "\emph{is the transitive closure of}".

We extend the usual Tarski semantics to third-order logic, \ie, interpretations $\calI$ map $\tau$-variables to $\tau$-objects.
An interpretation $\calI$ satisfies $\fraka(A_1,\ldots,A_k)$ if $(\calI(A_1),\ldots,\calI(A_k)) \in \calI(\fraka)$.
The quantifiers work as expected.

The set of all formulas of third-order arithmetic, that is, third-order formulas over the vocabulary $(+,\times,0,1,=,\leq)$, is written $\Delta^3_0$.
Likewise, the set of formulas of second-order arithmetic is $\Delta^2_0$.
The subset of closed formulas that are true in $\N$ is denoted by a boldface letter, \ie, $\mathbf{\Delta^3_0}$ and $\mathbf{\Delta^2_0}$, respectively.\footnote{With the notation, we follow the convention for first-order arithmetic $\Delta^1_0$ and second-order arithmetic $\Delta^2_0$; an equivalent notation for $\Delta^3_0$ would be $\Delta^2_\omega$~\cite{0066920}.
}

\subsection{The uncountable cases}

We reduce $\LTL(\sneg)$ to third-order arithmetic as follows.
The propositions $p_1,p_2,\ldots$ are identified with numbers $1,2,\ldots$.
The idea is now that a trace $t$ can be encoded as a binary relation $S$ such that $S(j,k)$ is true iff $p_k \in t(j)$, and vice versa.
Finally, a team is encoded as a unary third-order relation $\calA$ of binary relations $S$ such that $S \in \calA$ iff $S$ represents some trace $t \in T$.
Now, we define a translation $\rho_\fraka(\phi) \in \Delta^3_0$ of the $\LTL(\sneg)$-formula $\phi$ with one free third-order variable $\fraka$ of type $(2)$.
The translation is faithful in the sense that a team $T$ satisfies $\phi$ if and only if $(\N,\calA) \vDash \rho_\fraka(\phi)$, where $\calA$ is the third-order relation encoded by $T$.

In what follows, we restrict ourselves to the temporal operators $\X$ and $\U$, since the remaining ones are expressible as $\G \phi \equiv \sneg \F \sneg \phi$, $\F \phi \equiv \top \U \phi$ and $\phi \R \psi \equiv \sneg (\sneg \phi \U \sneg \psi)$.
Moreover, we can assume that $\neg$ occurs only in front of atomic propositional formulas, by an argument as in the proof of \Cref{thm:stutter-invariance-ltl}.

Let $\underline{m}$ denote the term $\underbrace{1+1+\cdots+1}_{m\text{ times}}$ for $m \geq 1$, or $0$ if $m = 0$, respectively.

The atomic formulas and Boolean connectives are straightforward:
\begin{align*}
  \rho_\fraka(p_k) &\dfn \forall S (\fraka(S) \imp S(0,\underline{k}))\\
  \rho_\fraka(\neg p_k) & \dfn \forall S (\fraka(S) \imp \neg S(0,\underline{k}))\\
  \rho_\fraka(\psi \land \theta) & \dfn \rho_\fraka(\psi) \land \rho_\fraka(\theta)\\
  \rho_\fraka(\sneg \psi) & \dfn \neg \rho_\fraka(\psi)
\intertext{For the splitting connective, we existentially quantify two subteams:}
  \rho_\fraka(\psi \lor \theta) & \dfn \exists \frakb \exists \frakc \big( \forall S (\fraka(S) \equi (\frakb(S) \lor \frakc(S))) \land \rho_\frakb(\psi) \land \rho_\frakc(\theta) \big)
\intertext{For the temporal operators, we define an auxiliary formula $\fraka \xrightarrow{d} \frakb$ with a free first-order variable $d$ and two free third-order variables $\fraka$ and $\frakb$, stating that $\frakb$ represents the team $T^d$ when $\fraka$ represents the team $T$:}
  \fraka \xrightarrow{d} \frakb & \dfn \forall S \big( \frakb(S) \equi \exists S' (\fraka(S') \land \forall j \forall k (S(j,k) \equi S'(j+d,k)))    \big)
  \intertext{Then we can translate $\X$ and $\U$ as follows:}
  \rho_{\fraka}(\X \psi) & \dfn \exists \frakb (\fraka \xrightarrow{1} \frakb \land \rho_\frakb(\psi))\\
 \rho_\fraka(\psi \U \theta) & \dfn \exists d \, \forall e \, \exists \frakb \, \exists \frakc \big(\fraka \xrightarrow{d} \frakb \land \fraka \xrightarrow{e}\frakc \land \rho_\frakb(\theta) \land (e < d \imp \rho_{\frakc}(\psi))\big)
\end{align*}

\begin{lemma}\label{lem:delta3-equiv}
 Let $\phi \in \LTL(\sneg)$ and let $T$ be a team.
 Let
 \[
  \calA_T \dfn \Big\{ \big\{ (j,k) \in \N^2 \mid p_k \in t(j) \big\} \mid t \in T\Big\}\text{.}
 \]
 Then $T \vDash \phi$ if and only if $(\N, \calA_T) \vDash \rho_\fraka(\phi)$.
\end{lemma}
\begin{proof}
Straightforward induction on the syntax of $\phi$.
\end{proof}

\begin{theorem}\label{thm:ltlsat-to-delta3}
Satisfiability and finite satisfiability of $\LTL(\sneg)$ are logspace-reducible to $\mathbf{\Delta^3_0}$.
\end{theorem}
\begin{proof}
 We apply the previous lemma.
 Assume $\phi \in \LTL(\sneg)$.
Clearly, by the above lemma, $\phi$ is satisfiable if and only if $\exists \fraka \, \rho_\fraka(\phi)$ holds in $\N$.

  The case where the team is finitely generated is more complicated.
  Suppose that $T = T(\calK)$, where $\calK = (W,R,\eta,r)$ is a finite structure.
  Essentially, the idea is to first quantify the components of $\calK$ as second-order objects, and then to express in the logic that we are looking precisely at the traces in $T(\calK)$.
For this, \wloss $W \subseteq \N$, $R \subseteq \N \times \N$ and $r = 0$.
  The propositional assignment $\eta$ can equivalently be represented as a relation $\hat{\eta} \dfn \{ (n,k) \mid p_k \in \eta(n) \}$.
We can then talk about $\calK$ inside arithmetic.
 First, the set $W$ should be finite and non-empty.
 Also, $R$ should be total and a subset of $W \times W$:
  \begin{align*}
    \psi_{\mathrm{frame}}(W,R) &\dfn W(0) \land \exists n \, \forall m \, (W(m) \to m < n) \land \forall m \exists n R(m,n)\\
    & \qquad \qquad \land\, \forall m \forall n (R(m,n) \to (W(m) \land W(n)))
  \end{align*}
 A path $\pi$ through $\calK$ is then simply a function $\N \to W$, hence a second-order object.
 We assert that $\pi$ is an $R$-path starting at the root,
 \begin{align*}
   \psi_{\mathrm{path}}(W,R,\pi) & \dfn \pi(0) = 0 \land \forall j \, \big(W(\pi(j))\land R(\pi(j),\pi(j+1))  \big)\text{.}
 \end{align*}
 We model a trace $t$ as the relation $S \dfn \{ (j,k) \mid p_k \in t_j \}$.
We can state that a trace is in $T(\calK)$ by saying that it is the trace of some path from $w$:
 \begin{align*}
   \psi_{\mathrm{trace}}(W,R,\hat{\eta},S) & \dfn \exists \pi \big( \psi_{\mathrm{path}}(W,R,\pi) \land \forall j \forall k (S(j,k) \leftrightarrow \hat{\eta}(\pi(j),k)) \big)
 \end{align*}
The formula
  \begin{align*}
   \psi_{\mathrm{generated}}(W,R,\hat{\eta},\fraka) \dfn \; & \forall S (\fraka(S) \equi \psi_{\mathrm{trace}}(W,R,\hat{\eta},S))
  \end{align*}
  expresses that $\fraka$ contains precisely all traces in $T(\calK)$.
Now $\phi$ is satisfied by the team $T(\calK)$, for some finite $\calK$,
if and only if
 \begin{align*}
\N \vDash \exists W \exists R \exists \hat{\eta} \exists \fraka \, \big( \psi_{\mathrm{frame}}(W,R) \land \psi_{\mathrm{generated}}(W,R,\hat{\eta},\fraka) \land \rho_{\fraka}(\phi) \big)
 \end{align*}

 The formula constructed in the above reduction consists of the inductive translation of $\phi$ as well as a constant part.
 The terms $\underline{k}$ have length linear in $k$, but \wloss $k$ is bounded by $\size{\phi}$.
 For this reason, it is straightforward to show that $\rho_{\fraka}(\phi)$, and hence the whole formula, is computable in logarithmic space.
\end{proof}

In the remainder of this section, we will present similar reductions, which all are logspace-computable as well.

\begin{theorem}\label{thm:ltlmc-to-delta3}
Model checking of $\LTL(\sneg)$ is logspace-reducible to $\mathbf{\Delta^3_0}$.
\end{theorem}
\begin{proof}
  We proceed as in the proof for satisfiability restricted to finitely generated teams, but additionally have to claim on the level of formulas that the relations $W$, $R$, and $\hat{\eta}$, which are quantified in the logic, are equal to the structure which is given as the input instance, say, $(W', R', \eta', r', \phi)$.
\Wloss the input structure is of the form $W' = \{0,1,\ldots,m-1\}$, for some $m > 0$, and $r' = 0$.

  Then the conjunction of the formulas
  \begin{align*}
    \psi_{=W'}(W) &\dfn \forall n (W(n) \equi n < \underline{m})\\
\psi_{=R'}(R) &\dfn \forall n \forall m (R(n,m) \equi \bigvee_{\mathclap{(i,j) \in R}}(n = \underline{i} \land m = \underline{j}))\\
\psi_{=\eta'}(\hat{\eta}) & \dfn \forall n \forall k (\hat{\eta}(n,k) \equi \bigvee_{\mathclap{\substack{i < m\\p_j \in \eta'(i)}}} (n = \underline{i} \land k = \underline{j})
\end{align*}
asserts that $(W', R', \eta') = (W, R, \eta)$.
Hence in the reduction of the previous theorem, we can replace the subformula $\psi_{\mathrm{generated}}(W,R,\hat{\eta},\fraka)$ by
\begin{align*}
\psi_{\mathrm{generated}}(W,R,\hat{\eta},\fraka) \land \psi_{=W'}(W) \land \psi_{=R'}(R) \land \psi_{=\eta'}(\hat{\eta})
\end{align*}

which proves the theorem.
\end{proof}

\subsection{The countable cases}

Next, we proceed with the decision problems that are reducible to $\mathbf{\Delta^2_0}$, \ie, second-order arithmetic.
Here, we can only manage countable teams, since the interpretations of second-order variables are always countable objects.
Given such a team $T$, we can assume that it is of the form $\{ t_i\}_{i \in I}$ for some $I \subseteq \N$.

We encode $T$ as a pair $(I,P) \in \wp\N \times \wp\N^3$ (\emph{indices} and \emph{propositions}) such that $I \subseteq \N$ is as in the subscript of $\{t_i\}_{i \in I}$, and $P$ describes the traces in the sense that $P(i,j,k)$ is true iff $p_k \in t_i(j)$.

The temporal operators are implemented by "shifting" all entries in $P$ by an offset $d$.
This works similarly to the $\Delta^3_0$ case.
The set $I$ is unaffected by this.
The atomic formulas and Boolean connectives are again straightforward:
\begin{align*}
  \rho_{I,P}(p_k) &\dfn \forall i (I(i) \imp P(i,0,\underline{k}))\\
  \rho_{I,P}(\neg p_k) & \dfn \forall i (I(i) \imp \neg P(i,0,\underline{k}))\\
  \rho_{I,P}(\psi \land \theta) & \dfn \rho_{I,P}(\psi) \land \rho_{I,P}(\theta)\\
  \rho_{I,P}(\sneg \psi) & \dfn \neg \rho_{I,P}(\psi)
\intertext{For the splitting connective, we need to divide $I$ only:}
  \rho_{I,P}(\psi \lor \theta) & \dfn \exists I' \exists I'' \big( \forall i (I(i) \equi (I'(i) \lor I''(i))) \land \rho_{I',P}(\psi) \land \rho_{I'',P}(\theta) \big)
\intertext{For the temporal operators, we again define an auxiliary formula:}
  P \xrightarrow{d} Q & \dfn \forall i \forall j \forall k \big( Q(i,j,k) \equi P(i,j+d,k) \big)\\
  \rho_{I,P}(\X \psi) & \dfn \exists P' (P \xrightarrow{1} P' \land \rho_{I,P'}(\psi))\\
 \rho_{I,P}(\psi \U \theta) & \dfn \exists d \, \forall e \, \exists P' \, \exists P'' \big(P \xrightarrow{d} P' \land P \xrightarrow{e}P'' \\
 & \qquad \qquad \land \rho_{I,P'}(\theta) \land (e < d \imp \rho_{I,P''}(\psi))\big)
\end{align*}

\begin{lemma}\label{lem:delta2-equiv}
 Let $\phi$ be a formula and $T = \{t_i\}_{i \in I}$ a team, where $I \subseteq \N$.
 Let $P \dfn \{ (i,j,k) \mid i \in I \text{ and } p_k \in t_i(j) \}$.
Then $T \vDash \phi$ if and only if $(\N, I, P) \vDash \rho_{I,P}(\phi)$ is true.
\end{lemma}
\begin{proof}
  Easy induction on the syntax of $\phi$ analogous to \Cref{lem:delta3-equiv}.
\end{proof}

\begin{theorem}\label{thm:ltlsat-to-delta2}
  Countable satisfiability, $\calC$-restricted satisfiability and $\calC$-restricted finite satisfiability of $\LTL(\sneg)$ are logspace-reducible to $\mathbf{\Delta^2_0}$ if $\calC$ is the class of ultimately periodic traces or the class of ultimately constant traces.
\end{theorem}

Note that we do not consider countable finite satisfiability, as it coincides with finite satisfiability over ultimately periodic teams by \Cref{prop:fin-count-is-ulp}.

\begin{proof}
A formula $\phi$ is satisfied in a countable team if and only if $\N \vDash \exists I \exists P \, \rho_{I,P}(\phi)$.
This immediately follows from \Cref{lem:delta2-equiv}.
Moreover, the trace number $i$ represented in $P$ is ultimately periodic if
 $(\N,i,P) \vDash \psi_{\mathrm{ulp}}(i,P)$, where
 \begin{align*}
  \psi_{\ulp}(i,P) \dfn \; & \exists c \, \exists d (d > 0 \land \forall j \forall k \, (j \geq c \imp (P(i,j,k) \equi P(i,j+d,k)))))\text{.}
 \end{align*}
It follows that
 \begin{align*}
   \N \vDash \exists I \, \exists P\, \big(\rho_{I,P}(\phi) \land \forall i (I(i) \imp \psi_{\ulp}(i,P)) \big)
 \end{align*}
 iff $\phi$ is satisfied by some team of ultimately periodic traces.
 For ultimately constant traces we simply replace $d$ in $\psi_{\ulp}$ by $1$.

\smallskip

 We proceed with the finitely satisfiable cases.
 For this, we use the formulas $\psi_{\mathrm{generated}}$ and $\psi_{\mathrm{frame}}$ from the proof of \Cref{thm:ltlsat-to-delta3}, but now with a pair $(I,P)$ as argument instead of a higher-order relation.
 More precisely, the formula
 \begin{align*}
  \psi_{\mathrm{generated,ulp}}(W,R,\hat{\eta},I,P) \dfn \; & \forall S \Big(\Big(\psi_{\mathrm{trace}}(W,R,\hat{\eta},S)\\
  \land \, \exists c \exists d \big(d > 0 \land \forall j \forall k  &(j\geq c \to (S(j,k) \equi S(j+d,k)) )  \big) \Big) \\
  & \qquad \equi \Big(\exists i \, \forall j \, \forall k (S(j,k) \equi P(i,j,k)) \Big)\Big)
 \end{align*}
stores all ultimately periodic traces in $(I,P)$, with $\psi_{\mathrm{trace}}$ as in \Cref{thm:ltlsat-to-delta3}.
Then
\begin{align*}
\hspace{-2.2ex}\N \vDash \exists W\, \exists R\, \exists \hat{\eta}\, \exists I \, \exists P\,  \Big( \psi_{\mathrm{frame}}(W,R) \land \psi_{\mathrm{generated,ulp}}(W,R,\hat{\eta},I,P)  \land  \rho_{I,P}(\phi) \Big)
\end{align*}
iff $\phi$ is satisfied in the ultimately periodic traces of a finitely generated team.
Again, the proof for the ultimately constant case is similar.
\end{proof}

\begin{theorem}\label{thm:countable-upper-mc}
  Ultimately periodic and ultimately constant model checking of $\LTL(\sneg)$ are reducible to $\mathbf{\Delta^2_0}$.
\end{theorem}
\begin{proof}
 Completely analogous to \Cref{thm:ltlmc-to-delta3} and \ref{thm:ltlsat-to-delta2}.
 The formula is
 \begin{align*}
   \exists W\, \exists R\, \exists \hat{\eta}\, &\exists I \, \exists P\, \exists I \, \big(\psi_{\mathrm{frame}}(W,R) \land \psi_{\mathrm{generated,ulp}}(W,R,\hat{\eta},I,P) \\
   & \qquad \land \psi_{=W'}(W) \land \psi_{=R'}(R) \land \psi_{=\eta'}(\hat{\eta}) \land \rho_{I,P}(\phi)\big)\text{,}
 \end{align*}
 with the ultimately constant case again similar.
\end{proof}
 \section{From third-order arithmetic to model checking}\label{sec:mc}

In this section, we make the next step to prove \Cref{thm:uncountable-main-thm} and \ref{thm:countable-main-thm}: we state the lower bounds of the model checking problem.
For this, a given $\Delta^3_0$-formula $\phi$ is translated to an $\LTL(\sneg)$-formula $\psi$ and a structure $\calK$ such that $\N \vDash \phi \LR T(\calK) \vDash \psi$.
In the first subsection, we begin with some preprocessing on $\phi$.
Mainly, we reduce the maximal arity of quantified relations, which simplifies the subsequent steps significantly.

Hence, let $\phi \in \Delta^3_0$ be closed.
First, we bring $\phi$ into prenex form by a routine transformation.
So \wloss $\phi = Q_1 X_1\cdots Q_n X_n \theta$, where $n \geq 1$, $\theta$ is quantifier-free, $\{Q_1,\ldots,Q_n\}\subseteq \{\exists, \forall\}$, and $X_1,\ldots,X_n$ are pairwise distinct (first-order, second-order, or third-order) variables.

\subsection{A bounded-arity normal form of $\Delta^3_0$}

It is a well-known fact that there are first-order definable \emph{pairing functions}, \ie, bijections $\pi \colon\N \times \N \to \N$.
One example is the Cantor polynomial,
\begin{align*}
  \pi(n,m) \dfn \frac{1}{2}\big((n + m)^2 + 3n + m\big)\text{.}
\end{align*}
It can easily be generalized to arbitrary arities by
\begin{align*}
  \pi^{(\ell)}(n_1,\ldots,n_\ell) \dfn \pi(\pi^{(\ell-1)}(n_1,\ldots,n_{\ell-1}),n_\ell)\text{, }\pi^{(1)}(n) \dfn n\text{.}
\end{align*}
That this allows to reduce quantified relation symbols to unary ones is a standard result in second-order logic.
Here, we prove it for formulas with third-order atoms.

It is routine to simulate all quantified function symbols by relation symbols, so \wloss the only function symbols are $+, \times$ and the numerical constants.
Moreover, we can assume that the built-in symbols $+,\times,0,1,<$ do not occur inside higher-order atoms (otherwise we replace them by quantified copies).

We naturally extend the definition of $\pi$ to relations $A \subseteq \N^\ell$ by
$\pi(A) \dfn \{ \pi^{(\ell)}(\mathbf{n}) \mid \mathbf{n} \in A\}$.
Likewise, for higher-order relations $\fraka \subseteq \wp\N^{\ell_1} \times \cdots \times \wp\N^{\ell_k}$, let
\[
\pi(\fraka) \dfn \{ (\pi^{(\ell_1)}(A_1),\ldots,\pi^{(\ell_k)}(A_k)) \mid (A_1,\ldots,A_k) \in \fraka \}\text{.}
\]
Finally, for an interpretation $\calI$, we write $\pi(\calI)$ for the interpretation that agrees with $\calI$ on first-order variables and has $\pi(\calI)(X) = \pi(\calI(X))$ for each (second- or third-order) variable $X$.
In the interpretation $\pi(\calI)$, now all second-order variables are mapped to unary relations, \ie, sets, and all third-order variables are mapped to relations of type $(1,\ldots,1)$.
\begin{lemma}
  Let $f_\ell$ be a fixed $\ell$-ary function variable not occurring in $\phi$.
  Let $\phi'$ be obtained from $\phi$ by replacing each atomic formula $A(t_1,\ldots,t_\ell)$, where $A$ is a second-order variable, by $A(f_\ell(t_1,\ldots,t_\ell))$, and furthermore changing the arity of each second-order variable to one, and changing the type of each third-order variable to $(1,\ldots,1)$.
  Then $(\N,\calI) \vDash \phi \LR (\N,\pi(\calI)) \vDash \phi'$ for all interpretations $\calI$ in which $\calI(f_\ell) = \pi^{(\ell)}$.
\end{lemma}

\begin{proof}
  By induction on $\phi$.
  The only interesting cases are the following:
\begin{itemize}
  \item If $\phi = A(t_1,\ldots,t_\ell)$ is atomic, with $A$ a second-order variable and $t_1,\ldots,t_\ell$ first-order terms, then the equivalence holds by definition, as
  \begin{align*}
    (\N,\calI) \vDash A(t_1,\ldots,t_\ell) &\LR (\calI(t_1),\ldots,\calI(t_\ell)) \in \calI(A)\\ &\LR \pi^{(\ell)}(\calI(t_1),\ldots,\calI(t_\ell)) \in \pi(\calI(A))\\
  & \LR \calI(f_\ell(t_1,\ldots,t_\ell)) \in \pi(\calI(A)) = \pi(\calI)(A)\\
  & \LR (\N,\pi(\calI)) \vDash A(f_\ell(t_1,\ldots,t_\ell))\text{.} \end{align*}
  \item If $\phi$ is atomic third-order, then again by definition
  $(\calI(A_1),\ldots,\calI(A_k)) \in \calI(\fraka) \LR (\pi(\calI(A_1)),\ldots,\pi(\calI(A_k))) \in \pi(\calI(\fraka)) = \pi(\calI)(\fraka)$.
  \item If $\phi = \exists A\, \psi$ and $A$ is $\ell$-ary second-order, then this follows from the fact that $\pi \colon \wp(\N^\ell) \to \wp(\N)$ as defined above is a bijection and by induction hypothesis.
\item If $\phi = \exists \fraka \,\psi$ and $\fraka$ is third-order, then the argument is similar.\qedhere
\end{itemize}
\end{proof}

Next, we aim at eliminating the newly introduced function $f_\ell$.
The graph of the pairing function $\pi^{(\ell)}$ is definable by a formula $\psi_\ell$ with $\ell+1$ arguments, viz.
\begin{align*}
  \psi_\ell(t_1,\ldots,t_\ell,t) &\dfn \exists x_2 \cdots \exists x_{\ell-1} (\psi_2(t_1,t_2,x_2) \land \cdots \land \psi_2(x_{\ell-1},t_\ell,t))
  \intertext{where}
  \psi_2(t_1,t_2,t) & \dfn (\underline{2} \times t) = ((t_1 + t_2)\times(t_1 + t_2)) + (\underline{3}\times t_1) + t_2
\end{align*}
defines the Cantor polynomial.
By this, the formula $A(f_\ell(t_1,\ldots,t_\ell))$ can equivalently be translated to $\exists x (\psi(t_1,\ldots,t_\ell,x) \land A(x))$.
As a consequence, we can assume that all second-order variables are unary.

Next, we reduce the arity of third-order variables.
As they now all have type $(1,\ldots,1)$ this is straightforward; a suitable pairing function $\Pi \colon (\wp\N)^k \to \wp \N$ is
\begin{align*}
  \Pi(A_1, \ldots, A_k) \dfn \bigcup_{i \in [k]} \{ k \cdot n + (i-1) \mid n \in A_i \}\text{,}
\end{align*}
where (the graph of) $\Pi$ is defined by
\begin{align*}
  \theta_k(A_1,\ldots,A_k,B) \dfn \forall m \big(B(m) \equi \exists n \bigvee_{\mathclap{i \in [k]}}(A_i(n) \land m = \underline{k} \cdot n + \underline{i-1}))\big)\text{.}
\end{align*}
On the level of formulas, we replace $\fraka(A_1,\ldots,A_k)$ with $\exists B (\theta_k(A_1,\ldots,A_k,B) \land \fraka(B))$, and make all third-order variables unary, analogously to the second-order case.

Finally, we can assume that the only built-in non-logical symbol is $<$, since $\leq$, $=$, $+$, $\times$ and all numerical constants are easily definable from it as quantified relations.
Observe that this re-introduces binary and ternary relation symbols, but ultimately, we have a constant bound of three on the arity.
In fact, we could either reduce the arity of \emph{all} relations to one and keep $+$ and $\times$ to express the pairing function, or eliminate $+$ and $\times$ but keep relations of arity $\ell > 1$.
But at least in second-order logic it is impossible to achieve both simultaneously, as the $\MSO(<)$-theory of $\N$ is decidable, known as Büchi's theorem \cite{buchi_mso}.
For third-order logic, to the best of the author's knowledge, this is open.

\begin{corollary}\label{cor:delta3-normal-form}
For every closed formula $\phi \in \Delta^3_0$ there is a logspace-computable closed formula $\psi \in \Delta^3_0$ such that $\N \vDash \phi \LR \N \vDash \psi$ and furthermore,
\begin{enumerate}
  \item $\psi$ is in prenex form, \ie, of the form $Q_1 X_1 \cdots Q_n X_n \, \theta$, where $\theta$ is quantifier-free, $\{Q_1,\ldots,Q_n\}\subseteq \{\exists, \forall\}$, and $X_1,\ldots,X_n$ are pairwise distinct (first-order, second-order, or third-order) variables.
  \item All atomic formulas in $\psi$ are of the form
    \begin{itemize}
      \item $\fraka(A)$, with $\fraka$ unary third-order, $A$ unary second-order, and $\{\fraka,A\} \subseteq \{X_1,\ldots,X_n\}$,
\item $A(x_1,\ldots,x_\ell)$, with $\ell \in \{1,2,3\}$, $A$ $\ell$-ary second-order, $x_1,\ldots,x_\ell$ first-order and $\{A,x_1,\ldots,x_\ell\}\subseteq \{X_1,\ldots,X_n\}$,
      \item or $x_1 < x_2$ with $\{ x_1,x_2 \} \subseteq \{X_1, \ldots, X_n \}$.
    \end{itemize}
\end{enumerate}
\end{corollary}

\subsection{Representing numbers and relations in traces and teams}

Next, we draw the connection to $\LTL(\sneg)$.
The crucial idea is that ($\ell$-tuples of) numbers, as well as sets thereof, can be encoded on traces.
For this, we use propositions $\Sigma \dfn \{ 0, 1 \}$.
Since a trace consists of countably infinitely many positions in a well-ordered fashion, it is natural to represent a number $n$ by simply setting a bit on the $n$-th position.
These are the traces generated by the structure shown in \Cref{fig:gadget-number}, aside from the single trace that never reaches $1$.
In other words, the structure generates all traces of the form $\emptyset\{0\}^*\{1\}\{0,\vend\}^\omega$ and $\emptyset\{0\}^\omega$.
Note that our encoding does not count the initial state of a trace, because this is fixed by the structure (and in our case labeled with $\emptyset$).

For relations, we simply set more bits to 1, \ie, a trace models a subset of $\N$ by setting a bit on every corresponding position.
These traces are generated by the structure in \Cref{fig:gadget-relation}.

Formally, a trace $t$ now represents the number $n \in \N$ if $\restr{t^1}{\Sigma} = \{0\}^{n}\{1\}\{0\}^\omega$.
The special proposition $\vend$ marks that 1 has been seen earlier on the trace.
This will be necessary later in the reduction.
For sets $A \subseteq \N$, a trace now represents $A$ if it holds for all $n\geq 0$ that $1 \in t(n+1)$ iff $n \in A$.

To account for $\ell$-tuples of numbers, where $\ell \in \{2,3\}$, we introduce copies $\Sigma_k \dfn \{ 0_k, 1_k \}$ of $\Sigma$, where $k \in \{1,2,3\}$.
The propositions $0_1$ and $1_1$ are identified with $0$ and $1$.
A trace $t$ now represents the $\ell$-tuple $(n_1,\ldots,n_\ell)$ if $\restr{t}{\Sigma_k}$ represents $n_k$ for all $k \in [\ell]$.
A structure generating all $\ell$-tuples is obtained from taking the $\ell$-fold product of that in \Cref{fig:gadget-number} and labeling the propositions accordingly, see \Cref{fig:gadget-tuple} for the case of $\ell = 2$.

To encode more complex objects, we require teams.
For binary or ternary second-order relations $A \subseteq \N^\ell$, a team $T$ represents $A$ if, for all tuples $\mathbf{n} = (n_1,\ldots,n_\ell) \in \N^\ell$, we have $\mathbf{n} \in A$ iff $\mathbf{n}$ is represented by some trace $t \in T$.
Finally, a team $T$ represents a third-order relation $\fraka \subseteq \wp \N$ if, for all $A \subseteq \N$, we have $A \in \fraka$ iff $A$ is represented by some trace $t \in T$.

Note that non-unary third-order relations or those with non-unary members could not feasibly be represented as set of traces, so the lengthy preprocessing of the previous subsection is crucial.

Let us stress that "trace-like" objects comprise numbers $n \in \N$, tuples $\mathbf{n} \in \N^\ell$, and sets $A \subseteq \N$, while "team-like" objects comprise sets of tuples $A \subseteq \N^\ell$ and higher-order sets $\fraka \subseteq \wp \N$.
Also, note that all possible logical atoms as in \Cref{cor:delta3-normal-form} boil down to comparing trace-like objects to each other, as well as checking membership of a trace-like object in a team-like object.
This crucially relies on the established normal form.

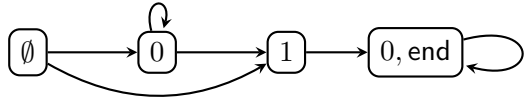
\begin{figure}\centering
  \begin{tikzpicture}[node distance=2mm,scale=.85,thick,->,>=stealth]
    \tikzset{world/.style={draw,rounded corners, rectangle,inner sep=1.5mm,black}}
    \tikzset{dummy/.style={inner sep=.6mm}}
    \tikzset{team/.style={draw,rounded corners,thick,inner sep=2mm}}

    \node[world] (root) at (0,0){$\emptyset$};
    \node[world] (end) at (6,0){$0,\vend$};

    \node[world] (0) at (2, 0) {$0$};
    \node[world] (1) at (4, 0) {$1$};

    \draw
    (root) edge (0)
    (root) edge[bend right] (1)
    (0) edge[loop above] (0)
    (0) edge (1)
    (1) edge (end)
    (end) edge[loop right] (end)
    ;
  \end{tikzpicture}
  \caption{Gadget where every number is represented by a trace $t$. The length of the first cycle between $\emptyset$ and $1$ determines the value, assuming $\restr{t(1)}{\{0,1\}} \in 0^*10^\omega$.}\label{fig:gadget-number}
\end{figure}

\begin{figure}\centering
  \begin{tikzpicture}[node distance=2mm,scale=.85,thick,->,>=stealth]
    \tikzset{world/.style={draw,rounded corners, rectangle,inner sep=1.5mm,black}}
    \tikzset{dummy/.style={inner sep=.6mm}}
    \tikzset{team/.style={draw,rounded corners,thick,inner sep=2mm}}

    \node[world] (root) at (0,0){$\emptyset$};

    \node[world] (0) at (3, 1) {$0$};
    \node[world] (1) at (3, -1) {$1$};

    \draw
    (root) edge (0)
    (root) edge (1)
    (0) edge[loop above] (0)
    (0) edge[<->] (1)
    (1) edge[loop below] (1)
    ;
  \end{tikzpicture}
  \caption{Gadget where every unary relation is represented by a trace. Visiting $1$ after $n$ steps means that the relation contains the number $n - 1$.}\label{fig:gadget-relation}
\end{figure}
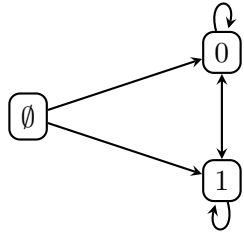

\begin{figure}\centering
  \begin{tikzpicture}[node distance=2mm,scale=.85,thick,->,>=stealth]
    \tikzset{world/.style={draw,rounded corners, rectangle,inner sep=1.5mm,black}}
    \tikzset{dummy/.style={inner sep=.6mm}}
    \tikzset{team/.style={draw,rounded corners,thick,inner sep=2mm}}

    \node[world] (root) at (0,-2){$\emptyset$};
    \node[world] (end) at (9, 0) {$0_1, 0_2, \vend$};

    \node[world] (00) at (0, 2) {$0_1,0_2$};
    \node[world] (01) at (3, 2) {$0_1,1_2$};
    \node[world] (11) at (3, 0) {$1_1,1_2$};
    \node[world] (10) at (3, -2) {$1_1,0_2$};

    \node[world] (00a) at (6, 2) {$0_1,0_2$};
    \node[world] (10a) at (9, 2) {$1_1,0_2$};

    \node[world] (00b) at (6, -2) {$0_1,0_2$};
    \node[world] (01b) at (9, -2) {$0_1,1_2$};

    \draw
    (root) edge (00)
    (00) edge[loop left] (00)
    (00) edge (01)
    (00) edge (10)
    (00) edge (11)
    (root) edge (01)
    (root) edge (10)
    (root) edge (11)
    (11) edge (end)
    (01) edge (00a)
    (01) edge[bend left] (10a)
    (00a) edge (10a)
    (00a) edge[loop below] (00a)
    (10a) edge (end)
    (10) edge (00b)
    (00b) edge[loop above] (00b)
    (10) edge[bend right] (01b)
    (00b) edge (01b)
    (01b) edge (end)
    (end) edge[loop right] (end)
    ;
  \end{tikzpicture}
  \caption{Gadget where every $2$-tuple is represented by a trace $t$ that is a "superposition" of two words of the form $\emptyset\{0\}^*\{1\}\{0\}^\omega$, that is, the projections $\restr{t}{\Sigma_1}$ and $\restr{t}{\Sigma_2}$ are such words.}\label{fig:gadget-tuple}
\end{figure}
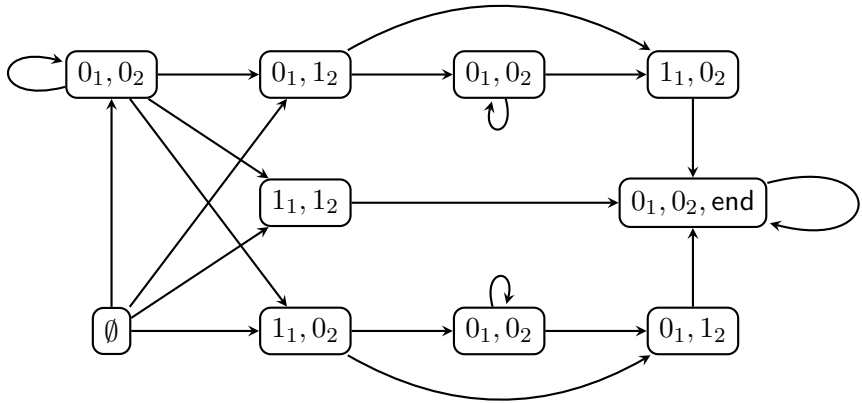

In the \Cref{fig:gadget-number,fig:gadget-relation,fig:gadget-tuple}, we saw structures that generate the corresponding traces for the encoding.
To potentially represent one number, tuple or relation for each variable $X_1,\ldots,X_n$, and to tell these apart, we meld together several instances $\calK_{X_i}$ of these structures, $i \in [n]$, depending on what kind of variable each $X_i$ is.

\begin{itemize}
\item If $X_i$ is first-order, and thus should represent a number, let $\calK_{X_i}$ be the structure shown in \Cref{fig:gadget-number}.
  We will quantify a single trace from it.
  \item If $X_i$ is either second-order or third-order, but unary, let $\calK_{X_i}$ be the structure shown in \Cref{fig:gadget-relation}.
  In the former case, we are interested in single traces, and in the latter case in sets of traces.
  \item Finally, if $X_i$ is second-order and of arity $\ell > 1$, then $\calK_{X_i}$ is constructed by taking the $\ell$-ary product of \Cref{fig:gadget-number} as demonstrated in \Cref{fig:gadget-tuple}.
\end{itemize}

Given a formula $\phi$ with variables $X_1,\ldots,X_n$, we now in general define the structure $\calK^{\phi}$\label{p:merged-structure} to be the disjoint union of the $\calK_{X_1},\ldots,\calK_{X_n}$, except that the roots of the $\calK_{X_i}$ are identified in $\calK$; call this world $r$.
Moreover, all non-root worlds from the respective $\calK_{X_i}$ are marked with the proposition $X_i$ (not shown in the figures), so as to cleanly separate the values represented for each $X_i$ in $\calK_{X_i}$.
With the proposition $X_i$ (or the formula $\F X_i$, if we are still in the root) we can determine whether a trace runs through the respective $\calK_{X_i}$.

A team $T \subseteq T(\calK^{\phi})$ now induces an interpretation $\calI_T$ of the variables $X_i$ as described above, depending on which traces of the respective $\calK_{X_i}$ are in $T$.
Recall that the notation $T_{\phi}$ refers to the subteam $\{ t \in T \mid \{t\} \vDash \phi\}$ of $T$.
\begin{itemize}
  \item If $X_i$ is first-order, $T_{\F X_i} = \{ t \}$ for some trace $t$, and $t \vDash \F \vend$, then $\calI_T(X_i)$ is the unique number $n$ such that $\restr{t^1}{\Sigma} = \{0\}^n\{1\}\{0\}^\omega$.
  \item If $X_i$ is second-order and unary, and $T_{\F X_i} = \{t \}$ for some trace $t$, then $\calI_T(X_i) \dfn \{ n \in \N \mid 1 \in t(n+1) \}$.
  \item If $X_i$ is second-order and binary or ternary, then $\calI_T(X_i) \dfn \{ (n_1,\ldots,n_\ell) \in \N^\ell \mid \exists t \in T_{\F X_i} : \forall j \in [\ell] : \restr{t^1}{\Sigma_j} = \{0\}^{n_j}\{1\}\{0\}^\omega \}$.
  \item If $X_i$ is third-order, then $\calI_T(X_i) \dfn \{ A \subseteq \N \mid \exists t \in T_{\F X_i} : t\text{ represents }A\}$.
\end{itemize}
Next, we explain how the team is manipulated in order to implement arithmetical quantifiers.
Suppose we start with the team $T = T(\calK^{\phi})$ and the outermost quantifier is $\exists X_1$.
The ideas is then to non-deterministically shrink the subteam $T_{X_1}$, where $X_1$ is one of the above cases, to a suitable subteam $U \subseteq T_{X_1}$ representing a value for $X_1$.
If $X_1$ is for example first-order, then $U$ should be a singleton.
Then we proceed with the team $(T \setminus T_{X_1}) \cup U$, and (existentially or universally) quantify an interpretation for $X_2$, and so on.

To confine the manipulation to a certain subteam $T_{X_i}$, we use "subteam quantifiers" $\qee{\phi}$ and $\qe{\phi}$ similar to $\qee{}$ and $\qe{}$ defined on p.~\pageref{p:quantifiers}.
Let $\phi \in \LTL$ and $\psi \in \LTL(\sneg)$, and let
\begin{align*}
  \qee{\phi} \psi &\dfn \neg\neg\phi \lor \psi\\
  \qe{\phi} \psi & \dfn \qee{\phi}((\qe{}\phi) \land \sneg \qee{\phi}((\qe{}\phi) \land \sneg \psi))\text{.}
\end{align*}
Then $\qee{\phi}$ intuitively says that we can shrink the subteam $T_{\phi}$ without touching the subteam $T_{\neg\phi}$.
Likewise, $\qe{\phi}$ says that we can shrink $T_{\phi}$ to a singleton.
The next lemma states this formally.
As before, the duals $\qaa{\phi}\psi \dfn \sneg \qee{\phi}\sneg\psi$ and $\qa{\phi}\psi \dfn \sneg \qe{\phi} \sneg \psi$ work as expected.

\begin{lemma}
A team $T$ satisfies $\qee{\phi} \psi$ if and only if there is a subteam $S \subseteq T_{\phi}$ such that $(T \setminus T_{\phi}) \cup S \vDash \psi$.
A team $T$ satisfies $\qe{\phi}\psi$ if and only if there is a trace $t \in T_{\phi}$ such that $(T \setminus T_{\phi}) \cup \{t\} \vDash \psi$.
\end{lemma}
\begin{proof}
  This was proved for so-called \emph{model team logic} in \cite[Proposition 5.3]{mtl_lueck}.
  For $\LTL(\sneg)$, the proof is identical.
\end{proof}

\subsection{Translating arithmetic formulas}

In this section, we translate $\Delta^3_0$-formulas into team logic.
First, let us repeat the definitions of the various non-classical connectives from team semantics (pp.~\pageref{p:team-connectives}--\pageref{p:hook}), as we need those in the remainder of the paper.
\begin{alignat*}{3}
  \intertext{Boolean connectives, including those definable from $\land$ and $\sneg$:}
  &T \vDash \phi \land \psi  && \; \LR \; T \vDash \phi \text{ and }T \vDash \psi\\
  &T \vDash \phi \ovee \psi && \; \LR \; T \vDash \phi \text{ or }T \vDash \psi\\
  &T \vDash \phi \timp \psi && \; \LR \;  T \vDash \phi \text{ implies } T \vDash \psi\\
  &T \vDash \phi \tequiv \psi &&\; \LR \;  T \vDash \phi \text{ iff }T \vDash \psi
  \intertext{Subteam connectives, including those definable from $\land,\lor,\sneg,\top,\bot$:}
  &T \vDash \neg \phi && \; \LR \; \forall t \in T : \{t\}\nvDash \phi\\
  &T \vDash \phi \lor \psi && \; \LR \; T = S \cup U \text{ such that }S \vDash \phi, U \vDash \psi\\
  &T \vDash \phi \hook \psi  && \; \LR \; T_\phi \vDash \psi\text{, where }T_\phi \dfn \{ t \in T \mid \{t\} \vDash \phi \}\\
  &T \vDash \qe{} \phi &&\; \LR \; \exists t \in T : \{t\} \vDash \phi\\
  &T \vDash \qa{} \phi && \; \LR \;  \forall t \in T : \{t\} \vDash \phi\\
  &T \vDash \qee{} \phi && \; \LR \;  \exists T' \subseteq T : T' \vDash \phi\\
  &T \vDash \qaa{} \phi && \; \LR \;  \forall T' \subseteq T : T' \vDash \phi\\
  &T \vDash \qe{\phi} \psi && \; \LR \;  \exists t \in T_\phi : (T \setminus T_{\phi}) \cup \{t\} \vDash \psi\\
  &T \vDash \qa{\phi} \psi && \; \LR \;  \forall t \in T_\phi : (T \setminus T_{\phi}) \cup \{t\} \vDash \psi\\
  &T \vDash \qee{\phi} \psi &&\; \LR \;  \exists T' \subseteq T_\phi : (T \setminus T_\phi) \cup T' \vDash \psi\\
  &T \vDash \qaa{\phi} \psi &&\; \LR \;  \forall T' \subseteq T_\phi : (T \setminus T_\phi) \cup T' \vDash \psi
\end{alignat*}

We are now in the position to state the inductive translation $\rho(\phi) \in \LTL(\sneg)$, where $\phi \in \Delta^3_0$.
Formally, $\rho$ should satisfy that $T \vDash \rho(\phi)$ iff $\calI_T \vDash \phi$.

We begin with the atomic formula $x_1 < x_2$, for which we have to compare two represented numbers $n_1$ and $n_2$, respectively.

It is straightforward to see that the formula
\begin{align*}
 \rho(x_1 < x_2) \dfn  \F ((x_1 \hook \vend) \land (x_2 \hook 1))
\end{align*}
implements $x_1 < x_2$, as it states that the digit $1$ on the trace of $n_2$ appears at some position where it already appeared beforehand (indicated by $\vend$) on the trace of $n_1$.
This clearly hinges on the fact that we have synchronous semantics.

Let us proceed with the atomic formula $A(x_1,\ldots,x_\ell)$, where $A$ is second-order.
Assume that $T_{x_i}$ represents a number $n_i$ for each $i \in [\ell]$, and that $T_A$ represents an $\ell$-ary relation, $\ell \in \{1,2,3\}$.
Then $(n_1,\ldots,n_\ell) \in \calI(A)$ iff $T$ satisfies
\begin{align*}
 \rho(A(x_1,\ldots,x_\ell)) \dfn \qe{\F A} \bigwedge_{j\in [\ell]} \F ((x_j \hook 1) \land (A \hook 1_j))\text{.}
\end{align*}
Intuitively, with $\qe{\F A}$ we select a trace $t$ in $\calK_A$ (in case $\ell = 1$ this has no effect, as unary relations are already encoded by single traces), and in the rest of the formula then check for each $j \in [\ell]$ that the 1 on the trace of $x_j$ appears on $\restr{t}{\Sigma_j}$ at the same position.

Finally, we consider the higher-order atom $\fraka(A)$.
Here, $A$ must be a unary relation symbol, so suppose $T_A$ is a single trace that represents a set of numbers.
Then $\fraka(A)$ is translated to
\begin{align*}
    \rho(\fraka(A)) \dfn \qe{\F \fraka} \G ((\fraka \hook 1) \tequiv (A \hook 1))\text{,}
\end{align*}
which selects a witness trace from $T_\fraka$ (representing a member of $\fraka$) and compares it to the single trace in $T_A$.
To compare two unary relations, we synchronously traverse the traces with $\G$ and check that the ones' positions coincide.

\medskip

After the atomic formulas, we now proceed with the remaining connectives.
The Boolean operators are straightforward:
$\rho(\psi_1 \land \psi_2) \dfn \rho(\psi_1) \land \rho(\psi_2)$ and $\rho(\neg\psi) \dfn \sneg \rho(\psi)$.

\smallskip

Finally, the quantifiers of arithmetic will be simulated by $\qe{}$ and $\qee{}$, where we have additional subformulas that ensure that the resulting subteam still represents a number or relation, respectively.
We can assume that all quantifiers are existential since $\forall X \psi \equiv \neg \exists X \neg \psi$.

For first-order $x$, let
\begin{align*}
  \rho(\exists x \, \psi) \dfn \qe{\F x} ((\F x \hook \F \vend) \land \rho(\psi))
\end{align*}
where $\F x \hook \F \vend$ excludes the trace that gets stuck in a loop of zeroes (cf.\ Figure~\ref{fig:gadget-number}).
For unary second-order $X$, we simply map
\begin{align*}
  \rho(\exists X \, \psi) \dfn \qe{\F X} \rho(\psi)
\end{align*}
since any trace in $T_{\F x}$ represents a unary relation and vice versa (cf.\ \Cref{fig:gadget-relation}).
For arity $\ell > 1$, recall that a relation is represented as a set of $\ell$-tuples, each sitting on its own trace (cf.\ \Cref{fig:gadget-tuple}).
Consequently,
\begin{align*}
  \rho(\exists X \, \psi) \dfn \qee{\F X} \big( (\F X \hook \neg\neg{}\F \vend) \land \rho(\psi) \big)\text{.}
\end{align*}
Here, with $\neg\neg{}\F \vend$ we say that no trace indefinitely avoids $\vend$, in other words, all traces represent some tuple.

Third-order relations are again easy,
\begin{align*}
  \rho(\exists \fraka \, \psi) \dfn \qee{\F \fraka} \rho(\psi)\text{,}
\end{align*}
since again any subteam of $T_\fraka$ represents a valid subset of $\wp \N$ and vice versa (cf.\ \Cref{fig:gadget-relation}).

In the next lemma, let $\calK^{\phi}$ again be the full structure defined as on p.~\pageref{p:merged-structure}.
\begin{lemma}\label{lem:reduction-base}
 $\N \vDash \phi$ iff $T(\calK^{\phi}) \vDash \rho(\phi)$.
\end{lemma}
\begin{proof}
 By straightforward induction.
\end{proof}

Note that all constructed formulas are expressible in $\LTL_1(\sneg,\F)$, \ie, with only the temporal operator $\F$ and without nesting of temporal operators.
In particular, observe that constructs such as $\qee{\F X} \psi$ and $\F X \hook \psi$ do not add to the nesting depth of $\psi$.

Since all the above formulas only involve standard recursion, they are easily shown logspace-computable.
Likewise, the structure $\calK^{\phi}$ is logspace-computable as it only consists of $n$ instances of constant substructures, where $n$ is the quantifier rank of $\phi$, with added propositions.
This concludes the main theorem of this section.

\begin{theorem}\label{thm:mc-hardness-delta3}
  $\mathbf{\Delta^3_0}$ is reducible to model checking of $\LTL_1(\sneg,\F)$.
\end{theorem}

In fact, this yields also lower bounds for some countable cases.
For this, observe that all gadgets except Figure~\ref{fig:gadget-relation} also work if we consider only ultimately periodic (or even ultimately constant) traces, that is, if every trace carries only finite information.
As a consequence, if we forbid third-order variables and also represent unary relations as (infinite) sets of traces representing $1$-tuples, then the reduction goes through and utilizes only ultimately constant traces.
This leads to the following result:

\begin{theorem}\label{thm:mc-hardness-delta2}
 $\mathbf{\Delta^2_0}$ is reducible to ultimately periodic model checking and ultimately constant model checking of $\LTL_1(\sneg,\F)$.
\end{theorem}

In the next section, we reduce model checking to satisfiability, and thus close the circle of logspace-reductions between these two problems and $\mathbf{\Delta^3_0}$ (resp.\ $\mathbf{\Delta^2_0}$).

\section{From Model Checking to Satisfiability}
\label{sec:sat}

A well-known feature of classical $\LTL$ is that its model checking problem can be reduced to satisfiability (cf.~Sistla and Clarke~\cite{SC85}).
The idea is, given a structure $\calK$, to encode it in a "characteristic formula" $\chi_\calK$ such that $t \vDash \chi_\calK$ if and only if $t$ is a trace in $\calK$.
As a consequence, $\chi_\calK \imp \phi$ is valid (that is, its negation unsatisfiable) if and only if $T(\calK) \vDash \phi$.

We elaborate a bit, following Schnoebelen~\cite{Sch02}.
Let $\calK = (W, R, \eta, r)$ be a structure over a finite set $\Phi \subseteq \AP$ of propositions.
\Wloss there are distinct propositions $p_w$ for every $w \in W$ such that $p_w \in \eta(w)$ and $p_w\notin \eta(w')$ for $w' \neq w$.
Then the formula $\chi_\calK$ is:\label{p:chi-k}
\begin{align*}
 \chi_\calK \dfn p_r \land \G \bigvee_{\mathclap{w \in W}}\Big(p_w \land \bigwedge_{\mathclap{\substack{w' \in W \\w' \neq w}}} \neg p_{w'} \land \bigwedge_{\mathclap{q \in \eta(w)}}q \; \land \; \bigwedge_{\mathclap{\substack{q \in \Phi\\q \notin \eta(w)}}}\neg q \;\land \bigvee_{\mathclap{(w,w') \in R}} \X p_{w'}\Big)
\end{align*}
Every trace in $\calK$ satisfies $\chi_\calK$, and
conversely, a trace that satisfies $\chi_\calK$ is in $\calK$.

\smallskip

An analogous construction for $\LTL(\sneg)$ would be a formula $\chi$ such that $T \vDash \chi$ if and only if $T = T(\calK)$.
However, in team-semantics, things become complicated:
Any classical formula, such as $\chi_\calK$, defines a downward closed class of traces, and as such it defines the \emph{subteams} of $T(\calK)$.
This is sufficient if $\phi$ itself is downward closed, and indeed, then $T(\calK) \vDash \phi$ iff $\neg \chi_\calK \lor \phi$ is valid, which reduces the model checking problem in team semantics to the validity problem.

If now $\phi$ itself is not downward closed, then we need some non-classical formula that requires the actual existence of traces of $\calK$ in the team (which is again not a downward closed property).

A formula $\phi$ \emph{defines a team $T$ (up to $\Phi$)} if it defines the class of teams $T'$ such that $\restr{T'}{\Phi} = T$.
Thus we need to define the team $T(\calK)$ up to $\Phi$, where $\Phi$ is the set of all propositions that occur in $\phi$.

\smallskip

Our approach works in two steps.
First we give a formula $\xi$ that defines the \emph{full} team $\frakT = (\wp \Phi)^\omega$, and then we use the formula $\chi_\calK$ to "weed out" traces not in $T(\calK)$.
For this, we use the fact that $\frakT_{\chi_{\calK}} = T(\calK)$;
recall that $T_\psi = \{t\in T \mid \{t\}\vDash \psi\}$.
Also recall that $T \vDash \phi_1 \hook \phi_2$ iff $T_{\phi_1}\vDash\phi_2$.
Then we obtain
\begin{align*}
 T(\calK) \vDash \phi  \; & \LR \; \frakT \vDash \chi_\calK \hook \phi\text{, and }\tag{$\star$}\\
\frakT \vDash \psi \;    & \LR \; \xi \land \psi \text{ is satisfiable }({\LR} \; \xi \timp \psi\text{ is valid),}\tag{$\star\star$}
\end{align*}
for all formulas $\phi,\psi$ that contain only propositions from $\Phi$.
Combining ($\star$) and ($\star\star$) yields a reduction from model checking to satisfiability.

It remains to construct the formula $\xi$ that defines $\frakT$.
We split this task into two steps, which each can be implemented in $\LTL(\sneg)$:
\begin{enumerate}
 \item Force all traces of the form $\emptyset^* \{ p \}\emptyset^\omega$ to appear in the team, where $p \in \Phi$.
\item For every subset $T$ of traces as in 1., force that the trace defined by $t(i) \dfn \bigcup_{t' \in T} t'(i)$ exists as well.
       As every trace can be expressed this way, this yields $\frakT$.\footnote{Readers familiar with HyperLTL will notice that this is why $\LTL(\sneg)$ can enforce uncountable teams and HyperLTL cannot.
       Roughly speaking, HyperLTL quantifies traces and binds them to trace variables, but cannot quantify and bind infinitely many traces at once.}
\end{enumerate}

To simplify the constructions, we make some refinements to this idea.
First, we introduce an auxiliary proposition $\# \notin \Phi$ to mark the positions \emph{after} $\{p\}$.
We also define the augmented team
\begin{align*}
  \frakT^\# \dfn \frakT \cup \{ \emptyset^n\{p\}\{\#\}^\omega \mid n \geq 0, p \in \Phi \}\text{.}
\end{align*}
Defining this team clearly suffices, as $\restr{\frakT^\#}{\Phi} = \frakT$.

In what follows, we refer to traces of the form $\emptyset^*\{p\}\{\#\}^\omega$ as \emph{prototraces}, as we build all other traces from these.
Prototraces are easily recognized in the team as they satisfy the LTL-formula $\F \#$, while "regular" traces do not.

The following $\LTL$-formula defines prototraces, as it states that exactly one proposition $p$ appears, that $\#$ must appear directly afterwards, and that from that position on, $\# \land \neg p$ holds forever.\label{p:def-proto}
\begin{align*}
 \xi_{\subseteq \mathrm{proto}} \dfn & \;\neg\neg\bigvee_{p \in \Phi} \Big(\F p \land \bigwedge_{\mathclap{p' \in \Phi\setminus \{p\}}}\G \neg p' \land \G(p \imp \X\#) \land \G(\# \imp \neg p\land \G\#)\Big)
 \intertext{Conversely, the non-classical formula}
 \xi_{\supseteq \mathrm{proto}} \dfn & \;\bigwedge_{\mathclap{p \in \Phi}} \G \qe{} p
\end{align*}
states for every $i \geq 0$ and $p \in \Phi$ that $p$ appears on some trace at position $i$.
Consequently, $\xi_{\mathrm{proto}} \dfn  \xi_{\subseteq\mathrm{proto}} \land \xi_{\supseteq\mathrm{proto}}$ means that the team contains precisely all prototraces.
Now we can quantify over sets of prototraces and state that their position-wise union appears as a trace, by saying that $p$ is false iff it is false in all prototraces:
\begin{align*}
 \xi \dfn &\; (\F\#\hook \xi_{\mathrm{proto}}) \land  \qaa{\F \#} \qe{\neg \F \#} \bigwedge_{\mathclap{p \in \Phi}} \G \big( (\neg\F \#\hook \neg p) \tequiv (\F \# \hook \neg p) \big)
\end{align*}\label{p:xi-formula}
Altogether, $\xi$ is the desired formula that defines $\frakT^\#$.\label{p:formula-xi}

\smallskip

A short inspection of the involved formulas reveals that they are all expressible in $\LTL_2(\sneg, \F, \X)$.
For this, again note that $\hook, \qa{}, \qe{}$ and so on do not increase temporal nesting, and that $\G$ is the same as $\sneg\F\sneg$.

\begin{theorem}\label{thm:mc-to-sat-delta3}
  Let $k \geq 2$.
Then the model checking problem of $\LTL_k(\sneg, \F, \X)$ is reducible to its satisfiability problem.
\end{theorem}

In fact, it is possible to eliminate $\X$ from the reduction (assuming that $\phi$ itself is an $\X$-free formula).
This is shown in the next subsection.

\subsection{A hard stutter-invariant fragment}

Let again a formula $\phi$ and an input structure $\calK = (W,R,\eta,r)$ be given, but now with $\phi \in \LTL_1(\sneg,\F)$.
In particular, $\phi$ is now stutter-invariant.
By \Cref{thm:mc-hardness-delta3}, this fragment of model checking is already as hard as the full problem.

First, we restate the first part of the reduction, that is, from model checking to truth in $\frakT^\#$, in the stutter-invariant fragment:
\begin{lemma}\label{lem:mc-f2-to-t}
  The model checking problem of $\LTL_1(\sneg,\F)$ is reducible to truth of $\LTL_2(\sneg,\F)$-formulas in $\frakT^\#$.
\end{lemma}

This requires eliminating the $\X$-operator in the formula $\chi_\calK$ (cf.~p.~\pageref{p:chi-k}).
The idea is to simulate it by using "helper" prototraces.
The following table illustrates this.

\medskip

\begin{center}
\begin{tabular}{CCCCCCC}
 \toprule
 t & \cdots & \emptyset & \{x\}     & \{\#\}    & \{\#\}  & \cdots \\
 t' & \cdots& \emptyset &  \emptyset & \{y\}     & \{\#\}  & \cdots \\
 t''  & \cdots& \eta(w)       & \eta(w')    & \eta(w'') & \eta(w''') & \cdots  \\
 \bottomrule
\end{tabular}
\end{center}

\medskip

We arbitrarily pick two distinct propositions $x,y \in \Phi$ (\wloss $\size{\Phi} \geq 2$).
Call a trace $t$ \emph{active at position $i$} if $t(i) \cap \Phi \neq \emptyset$.
We quantify two prototraces $t,t'$ such that $t$ is active with $x$ first, and $t'$ is active with $y$ directly after.
We say this by stating that $x$ and $y$ do not occur together, and also $\#$ may not occur simultaneously with $x$ or $\emptyset$, so $y$ must appear immediately when $x$ is gone.
\begin{align*}
  \zeta_1 \dfn & \; \F \# \hook \G\big(\sneg(x \lor y) \land \sneg(\# \lor x) \land \sneg(\# \lor (\neg x \land \neg y\land \neg \#) \big)
\intertext{This enables us to access consecutive positions in $t''$ by querying whether a prototrace is active with $x$ or $y$, respectively.
  Hence we can state that $x$ and $y$ appear together with propositions $p_w$ and $p_{w'}$, respectively, such that $w,w'$ are successors:}
  \zeta_2 \dfn &\; \bigvee_{\mathclap{(w,w')\in R}} \G((\F \# \hook \neg x) \ovee (\neg\F\# \hook p_w)) \land \G ((\F \# \hook \neg y) \ovee (\neg \F\# \hook p_{w'}))
\end{align*}
Let $\zeta \dfn  \,\qa{\neg\F\#} \, \qa{(\F\# \land \F x)} \, \qa{(\F\# \land \F y)} \, (\zeta_1 \timp \zeta_2)$.
Then $\zeta$ says that every non-prototrace contains a sequence $p_{w_1},p_{w_2},\ldots$ of variables only if $w_1,w_2,\ldots$ is a path in $\calK$.
The formula $\chi'_\calK$ that replaces $\chi_\calK$ can now be chosen as
\begin{align*}
  \chi'_\calK \dfn\bigg( &(\neg\F\#) \hook \neg\neg \Big( p_r \land \G \bigvee_{\mathclap{w \in W}}(p_w \land \bigwedge_{\mathclap{\substack{w' \in W \\w' \neq w}}} \neg p_{w'} \land \bigwedge_{\mathclap{q \in \eta(w)}}q \; \land \; \bigwedge_{\mathclap{q \notin \eta(w)}}\neg q ) \Big) \bigg)\land \zeta
\end{align*}

As before, we want to check $\phi$ only in $T(\calK)$, so we split the full team into $T(\calK)$ and $(\wp\Phi)^\omega \setminus T(\calK)$.
However, in contrast to before we cannot simply write $\chi_\calK \hook \phi$ anymore.
The reason is that both sides of the splitting now need access to prototraces, but due to the definition of $\hook$ (p.~\pageref{p:hook}), the subteam $\frakT^{\#}_{\chi_{\calK}}$ would not contain any.

To solve this, let $\chi''_\calK$ be a copy of the formula $\chi'_\calK$, but with the propositions $x$ and $y$ replaced by different distinct propositions $x'$ and $y'$, which are \wloss in $\Phi$.
One subteam of $\frakT^\#$ now contains all traces in $T(\calK)$, witnessed by $\chi'_\calK$ using prototraces for $x$ and $y$;
and the remaining subteam contains all traces in $(\wp \Phi)^\omega \setminus T(\calK)$, witnessed by prototraces for $x'$ and $y'$:\label{p:reduction-mc-to-t}
\begin{align*}
  T(\calK) \, \vDash \,\phi \; \LR \; \frakT^\# \vDash \;  & \Big((\F\# \hook \G (\neg x' \land\neg y')) \land \chi'_\calK \land (\neg \F \# \hook \phi)\Big) \\
   & \lor \Big((\F\# \hook \G (\neg x \land \neg y)) \land \sneg \chi''_\calK \Big)
\end{align*}
The above formula is a reduction from model checking to truth in $\frakT^\#$ for $\LTL_2(\sneg,\F)$-formulas.

\medskip

Next, we proceed with the second step by mapping the set of formulas true in $\frakT^\#$ to the satisfiability problem.
\begin{lemma}\label{lem:f2-t-to-sat}
  Truth of $\LTL_2(\sneg,\F)$-formulas in $\frakT^\#$ is reducible to the satisfiability problem of $\LTL_2(\sneg,\F)$.
\end{lemma}
This boils down to finding an $\LTL_2(\sneg,\F)$-formula that defines $\frakT^\#$ (up to stutter-equivalence).
We start with the formula that defines prototraces.
Without $\X$, we can only define traces of the form $\emptyset^*\{p\}^+\{\#\}^\omega$, which are stutter-equivalent to prototraces.
For this, the formula $\xi$ (cf.~p.~\pageref{p:def-proto}) is changed to $\xi'$ where the subformula $\G(p \to \X \#)$ is replaced with $\G(p \to (\F \# \land \G(p \lor \#)))$.

A problem arises when considering teams:
It may still be that the whole \emph{team} is not stutter-equivalent to a \emph{team} of prototraces, although every trace is, as the following example illustrates.
The team depicted below is stutter-free, but not stutter-equivalent to a team of prototraces.

\begin{center}
\begin{tabular}{CCCCCCC}
 \toprule
 t  & \emptyset & \{p\}     & \{p\}   & \{\#\}  & \{\#\} & \cdots \\
 t' & \emptyset & \emptyset & \{ q \} & \{ q \} & \{\#\} & \cdots \\
 \bottomrule
\end{tabular}
\end{center}
So we need to control the stuttering throughout the team.
The following formula stipulates that whenever prototraces overlap, in the sense that they have a common active position, their active positions are identical altogether (but not necessarily the labeled proposition).
We write $\bigvee\Phi$ short for $\bigvee_{p \in \Phi}p$.
\begin{align*}
  \psi_\mathrm{stutter,1} \dfn \qaa{\F\#} \Big(\#\F\hook \Big(( \F \bigvee \Phi) \timp \G ((\neg \# \land \neg \bigvee\Phi) \ovee \# \ovee \bigvee\Phi )\Big)\Big)
\end{align*}
With this, the situation depicted before cannot occur, and we obtain a team that is stutter-equivalent to a team of prototraces.
However, this is still not sufficient when other traces come into play, as the following example shows:
\begin{center}
\begin{tabular}{CCCCCCC}
 \toprule
 t  & \emptyset & \{p\}     & \{p\}   & \{\#\}  & \{\#\} & \cdots \\
 t' & \emptyset & \emptyset & \emptyset & \{ q \} & \{\#\} & \cdots \\
 t'' & \eta(w)  & \eta(w')  & \eta(w'') & \cdots  & \cdots & \cdots\\
 \bottomrule
\end{tabular}
\end{center}
The "regular" traces can still advance faster than the supposed prototraces, and as a consequence, the whole team is again not stutter-equivalent to $\frakT^\#$.
This can be remedied as follows.
We stipulate that, whenever a prototrace $t$ is active, no non-prototrace may change its label until $t$ switches to $\#$:
\begin{align*}
  \psi_\mathrm{stutter,2} \dfn \qa{\F\#} \qa{\neg\F\#} \bigwedge_{\mathclap{\substack{p \in \Phi\\\ell\in\{p,\neg p\}\\q\in \Phi}}}  \Big( \F\big((\F&\#\hook q) \land (\neg\F\# \hook \ell)\big)\\[-2.5em]
  & \qquad\timp \G \big( (\F\#\hook q) \timp (\neg\F\#\hook \ell)\big)\Big)
\end{align*}
Here, the first line asks whether there is some common position where a prototrace $t$ is active and a normal trace $t'$ satisfies some literal $\ell$.
If so, then the second line states that $t'$ must satisfy $\ell$ in \emph{all} positions where $t$ is active.
Let $\psi_\mathrm{stutter} \dfn \psi_\mathrm{stutter,1} \land \psi_\mathrm{stutter,2}$.

Now we have achieved that the whole team stutters whenever a prototrace stays active.
For this reason, every prototrace stays active for effectively only one time step.
Formally, for any team $T$, it holds that $T \vDash \xi' \land \psi_{\mathrm{stutter}}
$
if and only if $T \steq \frakT^\#$.
As a consequence, we obtain
\begin{align*}
  \frakT^\# \vDash \phi \;\LR \; & \exists T : T \steq \frakT^\# \text{ and }T\vDash \phi\tag{as $\phi$ is stutter-invariant}\\
  \LR \; & \exists T : T \vDash  \xi' \land \psi_{\mathrm{stutter}} \land \phi\tag{construction of $\xi'$ and $\psi_{\mathrm{stutter}}$}\\
  \LR \; &  \xi' \land \psi_{\mathrm{stutter}} \land \phi \text{ is satisfiable.}
\end{align*}

Again, all formulas have temporal depth at most two and use only temporal operators $\F$ and $\G$.
Combining \Cref{lem:mc-f2-to-t} and \Cref{lem:f2-t-to-sat} yields:

\begin{theorem}\label{thm:mc-to-sat-x-free}
  Let $k \geq 2$.
  Then the model checking problem of $\LTL_k(\sneg,\F)$ is reducible to its satisfiability problem.
\end{theorem}

\subsection{The countable cases}

We showed in \Cref{thm:countable-upper-mc} and \ref{thm:mc-hardness-delta2} that model checking of ultimately periodic/constant traces in a structure is equivalent to $\mathbf{\Delta^2_0}$.
Here, we transfer this result also to the satisfiability problem.
We focus on ultimately constant traces as our result will also cover the ultimately periodic case.

Let $\frakT_\ulc$ resp.\ $\frakT^\#_\ulc$ be the team of all ultimately constant traces in $\frakT$ resp.\ $\frakT^\#$.
Then, we need the following result, analogously to \Cref{lem:mc-f2-to-t}:
\begin{align*}
  T_\ulc(\calK) \vDash \phi \; &\LR \; \frakT_\ulc \vDash \chi_\calK \hook \phi
\end{align*}

But we need to address a subtle issue first.
For the formula $\chi_\calK$ to work, we assumed that each state $w$ of the structure has some proposition $p_w$ labeled uniquely in that state.
For general model checking, this was no loss of generality.
But unfortunately, adding such propositions changes the set of ultimately constant paths!
The following example illustrates this.

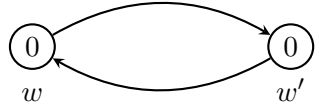
\begin{center}
  \begin{tikzpicture}[node distance=1mm,scale=.85,thick,->,>=stealth]
    \tikzset{world/.style={draw, circle,inner sep=1mm,black}}
    \tikzset{dummy/.style={inner sep=.6mm}}
    \node[world] (root) at (0,0){$0$};
    \node[world] (end) at (4,0){$0$};
    \node[below = of root] {$w$};
    \node[below = of end] {$w\smash{'}$};
    \draw
    (root) edge[bend left] (end)
    (end) edge[bend left] (root)
    ;
  \end{tikzpicture}
\end{center}

This structure has exactly one path from any state, which also induces an ultimately constant trace.
But adding propositions as mentioned before leads to the situation that the team of ultimately constant traces is empty, and certainly this is not a valid reduction.

\smallskip

To make the same reduction work, we need to ensure that ultimately constant traces are induced by paths through the structure that themselves are "ultimately constant", \ie, visit only one state infinitely often.
However, note that the reduction from $\mathbf{\Delta^2_0}$ in \Cref{sec:mc} uses precisely such structures.
In these, the ultimately constant traces are only induced by paths that get stuck in a loop, \ie, in an edge of a state to itself.
Therefore we can strengthen \Cref{thm:mc-hardness-delta2} and obtain the following lemma.

\begin{lemma}\label{lem:delta2-to-distinct-model-checking}
  $\mathbf{\Delta^2_0}$ is reducible to ultimately constant model checking of $\LTL_1(\sneg,\F)$ on structures where all states have pairwise distinct labels.
\end{lemma}

Thus we can assume the propositions $p_w$ for $w \in W$ labeled as before, and obtain
\begin{align*}
  T_\ulc(\calK) \vDash \phi \; &\LR \; \frakT_\ulc \vDash \chi_\calK \hook \phi\text{.}
\end{align*}

\smallskip

Next, we show that $\X$ can again be eliminated from $\chi_\calK$.
As $\frakT^\#_\ulc$ includes all prototraces (which are ultimately constant),
we can use the same formula as before on p.~\pageref{p:reduction-mc-to-t}:
\begin{align*}
  T_\ulc(\calK) \, \vDash \,\phi \; \LR \; \frakT^\#_\ulc \vDash \;  & \Big((\F\# \hook \G (\neg x' \land\neg y')) \land \chi'_\calK \land (\neg \F \# \hook \phi)\Big) \\
   & \lor \Big((\F\# \hook \G (\neg x \land \neg y)) \land \sneg \chi''_\calK \Big)
\end{align*}
This leads to the analogous reduction:

\begin{lemma}\label{lem:mc-delta2-to-t}
  Ultimately constant model checking of $\LTL_1(\sneg,\F)$ on structures where all states have pairwise distinct labels is reducible to truth of $\LTL_2(\sneg,\F)$ in $\frakT^\#_\ulc$.
\end{lemma}

It remains to reduce the truth problem of $\frakT^\#_\ulc$ to the satisfiability problem, as done before for $\frakT^\#$ in \Cref{lem:f2-t-to-sat}.
There, we showed that
\begin{align*}
  \frakT^\# \vDash \phi \;\LR\; (\xi' \land  \psi_{\mathrm{stutter}} \land \phi) \text{  is satisfiable,}
\end{align*}
where $\xi'$, $\psi_{\mathrm{stutter}}$ and $\phi$ all are $\LTL_2(\sneg,\F)$-formulas.
To define the team $\frakT^\#$ by a formula, we previously used the formula $\xi'$, in particular the subformula
\begin{align*}
  \qaa{\F\#} \qe{\neg \F\#} \bigwedge_{p \in \Phi} \G \big((\neg \F \# \hook \neg p) \tequiv (\F\# \hook \neg p)\big)\tag{$\star$}
\end{align*}
of $\xi'$ (cf.~p.~\pageref{p:xi-formula}) claimed that the position-wise union of every subset of prototraces appears as a regular trace in the team.
This now has to be restricted to ultimately constant traces.
Ultimately constant traces are defined by the $\LTL$-formula
\begin{align*}
  \alpha \dfn \bigwedge_{p \in \Phi} (\F\G p \lor \F \G \neg p)\text{.}
\end{align*}\label{p:def-of-alpha}
However, we want to express something slightly different:
that the union of the currently selected (ultimately constant) prototraces is still ultimately constant, which is not the case in general.
For this, we state that, in some suffix, $p \in \Phi$ either appears in no prototrace (that is, $\F\G\neg p$ holds), or that at every position, $p$ appears in some prototrace ($\F \G \sneg \neg p$ holds):
\begin{align*}
  \alpha' \dfn \F\#\hook \bigwedge_{p \in \Phi}(\F \G \neg p \ovee \F\G \sneg \neg p)
\end{align*}
Accordingly, we obtain the formula $\xi''$ by adding $\alpha'$ to the above subformula ($\star$) of $\xi'$:
\begin{align*}
  \qaa{\F\#} \Big( \alpha' \timp   \qe{\neg \F\#} \bigwedge_{p \in \Phi} \G \big((\neg \F \# \hook \neg p) \tequiv (\F\# \hook \neg p)\big)\Big)
\end{align*}
By this, we force only the ultimately constant traces to appear.
In total, we change the reduction to
\begin{align*}
  \phi \mapsto \xi'' \land \psi_{\mathrm{stutter}} \land \neg\neg\alpha \land \phi\text{.}
\end{align*}
Now $\xi'' \land \psi_{\mathrm{stutter}} \land \neg\neg\alpha$ defines $\frakT^\#_\ulc$ up to stuttering.
More precisely, a similar equivalence chain as in the proof for \Cref{thm:mc-to-sat-x-free} follows:

\begin{align*}
  \frakT^\#_\ulc \vDash \phi \;\LR \; & \frakT^\#_\ulc \vDash \neg\neg\alpha \land \phi  \tag{as $\neg\neg\alpha$ defines ultimately constant teams}\\
  \LR \; & \exists \, T : T \steq \frakT^\#_\ulc \text{ and }T\vDash\neg\neg\alpha\land \phi\tag{as $\neg\neg\alpha$ and $\phi$ are stutter-invariant}\\
  \LR \; & \exists T : T \vDash  \xi'' \land \psi_{\mathrm{stutter}} \land \neg\neg \alpha \land \phi  \tag{construction of $\xi''$ and $\psi_{\mathrm{stutter}}$}\\
  \LR \; &  \xi'' \land \psi_{\mathrm{stutter}} \land \neg\neg\alpha \land \phi \text{ is satisfiable.}
\end{align*}

\begin{lemma}\label{lem:t-to-sat-delta2}
Let $\calC$ contain at least all ultimately constant traces.
  Then the truth of $\LTL_2(\sneg,\F)$-formulas in $\frakT^\#_\ulc$ is reducible to the $\calC$-restricted satisfiability problem of $\LTL_2(\sneg,\F)$.
\end{lemma}

The combination of \Cref{lem:delta2-to-distinct-model-checking,lem:mc-delta2-to-t,lem:t-to-sat-delta2} yields:

\begin{theorem}\label{thm:arbitrary-delta2-sat}
  Let $\calC$ contain at least all ultimately constant traces.
Then $\mathbf{\Delta^2_0}$ is reducible to the $\calC$-restricted satisfiability problem of $\LTL_2(\sneg,\F)$.
\end{theorem}

In particular, the reduction in \Cref{lem:t-to-sat-delta2} produces formulas that are either unsatisfiable, or satisfied by some ultimately constant and hence countable team:

\begin{corollary}
  $\mathbf{\Delta^2_0}$ is reducible to the countable satisfiability problem of $\LTL_2(\sneg,\F)$.
\end{corollary}

\subsection{Finite satisfiability}

\label{subsec:generated}

Finally, we investigate the complexity of the problem of finite satisfiability, \ie, whether a formula is true in some team generated by a finite structure.
Unlike in classical $\LTL$, this is a genuinely different problem.
In $\LTL$, every satisfiable formula is satisfied by an ultimately periodic trace~\cite{SC85}, which itself is always finitely generated.
However, for example, the $\LTL(\sneg)$-formula
\begin{align*}
  (\G \qe{} (p \land \X \neg p )) \land \qa{}(p \U \G\neg p)
\end{align*}
defines the team $\{ \{p\}^n\emptyset^\omega \mid n \geq 1 \}$, which is not generated by any finite structure~\cite{abs-1907-05070}.

\smallskip

The full team $\frakT$ is not finitely generated for the simple reason that its traces differ on the first position.
The same holds for $\frakT^\#$, $\frakT_\ulc$ and $\frakT^\#_\ulc$.
However, avoiding this problem is not too difficult, we will show that these teams can actually be simulated by traces with a common root.
The idea is to "delay" the first label---we do not ask if a proposition $p$ is true at the beginning, but rather if there is some point in the future where $p^\leftarrow$ holds, where $p^\leftarrow$ is a fresh proposition that does not interact with the formula otherwise.

We add a new proposition $\vroot$ and claim that it is true in some non-empty, finite prefix of the team, consistently across all traces, and that nothing else is labeled simultaneously:
\begin{align*}
  \psi_{\mathrm{root}} \dfn \vroot \land \F \neg \vroot \land \G((\vroot \land \neg \# \land \bigwedge_{p \in \Phi}(\neg p \land \neg p^\leftarrow)) \ovee \G\neg \vroot)
\end{align*}
Modulo stuttering, this is of course equivalent to every trace $t$ having the initial label $t(0) = \{\vroot\}$ and $\vroot \notin t(1), t(2), \cdots$.
In the initial prefix, the value of any proposition $p \in \Phi$ on a trace is now simulated by the truth of $\F p^\leftarrow$.

We carefully replace $p$ in the formulas as follows in order to not increase the temporal depth to three.
\begin{itemize}
  \item In $\alpha$ and $\alpha'$, that is, the formulas stated on p.~\pageref{p:def-of-alpha} that say that a trace is ultimately constant, we change nothing, since disturbing $p$ on a finite prefix of a trace does not alter whether it is ultimately constant or not.
  \item The subformula $\G(p \to (\F \# \land \G (p \lor \#)))$ of $\xi'$, which states that on prototraces the proposition $p$ is followed by $\#$ (after possible stuttering of $p$), is replaced by $\G \big((p \lor \F p^\leftarrow) \to (\F \# \land \G(\vroot \lor p \lor \#))\big)$.
  \item Every other occurrence of $p \in \Phi$ has been only at temporal depth at most one, and can thus be replaced with $p \ovee (\vroot \land \neg\neg\F p^\leftarrow)$ without increasing the total temporal depth of two.
\end{itemize}

Afterwards, the formula $\psi_{\mathrm{root}}$ is appended to the reduction.

The team $\frakT^\#$ with the above changes is finitely generated:
Add arcs from a root state with label $\{\vroot\}$ to a fully connected set of states, one for each possible label, to generate $\frakT$.
For the prototraces, add a single path from the root that cycles through an empty label (or skips this altogether), eventually visits some $p \in \Phi$ and then cycles through $\{\#\}$.
Moreover, the ultimately constant traces in this structure form precisely the team $\frakT^\#_\ulc$.
By this, we can adapt \Cref{thm:arbitrary-delta2-sat} to the case of finitely generated teams.

\begin{theorem}
  Let $\calC$ contain at least all ultimately constant traces.
Then $\mathbf{\Delta^2_0}$ is reducible to the $\calC$-restricted finite satisfiability problem of $\LTL_2(\sneg,\F)$.
\end{theorem}

\section{Conclusion}

In this article, we studied the computational complexity of the logic we called $\LTL(\sneg$), \ie, LTL with synchronous team semantics and with Boolean negation $\sneg$.
We showed that both the model checking and the satisfiability problem are highly undecidable, each equivalent to $\mathbf{\Delta^3_0}$, that is, the set of all true formulas of third-order arithmetic.
The idea is that infinite traces can be seen as the characteristic sequences of subsets of $\N$, and teams hence represent sets of subsets of $\N$, which are third-order objects.
As one step in the hardness proof, we used the fact that $\LTL(\sneg)$ can enforce uncountable teams, which is possible in neither of the related logics HyperLTL~\cite{FZ17} and team-logical LTL without negation~\cite{KMV018}.
Over countable teams, for example when using only ultimately periodic traces, the complexity drops down to $\mathbf{\Delta^2_0}$, that is, "only" second-order arithmetic.

Several known features of the high expressivity of HyperLTL manifest in $\LTL(\sneg)$ as well.
For example, Finkbeiner and Zimmermann~\cite{FZ17} showed that HyperLTL can enforce aperiodic traces, which is a crucial step also for our reduction to $\LTL(\sneg)$ where infinite sets of numbers are identified with traces.

\smallskip

The lower complexity bounds already hold for weak fragments such as stutter-invariant $\LTL(\sneg)$, and in fact with only the future modality $\F$ available and temporal depth two.
For model checking, even temporal depth one suffices.
This deviates from classical LTL, where temporal depth at least two is required for the full hardness, for any combination of temporal operators~\cite{DS02,Sch02}.
The satisfiability problem of HyperLTL, however, is undecidable already for temporal depth one~\cite[Theorem 5]{abs-1907-05070}.
The question whether a similar result holds for $\LTL_1(\sneg)$ is open.

In future research, studying other weak fragments could yield further insight.
Besides $\LTL_1(\sneg)$, examples are the fragments using only $\X$ or with a bounded number of propositional variables, which has also been considered for classical LTL~\cite{DS02}.
Also, the asynchronous operators might be worth investigating.

With its high complexity, $\LTL(\sneg)$ is much harder than HyperLTL, which has a non-elementary but decidable model checking problem~\cite{CFKMRS14}.
In fact, it seems plausible that HyperLTL satisfiability is reducible to $\mathbf{\Delta^2_0}$, since every satisfiable formula has a countable model~\cite[Theorem 2]{FZ17}, which would again be easier than satisfiability of $\LTL(\sneg)$.

The lower bounds presented here heavily utilize the unrestricted Boolean negation $\sneg$ in team semantics in combination with team splitting.
A sensible restriction of negation may be a first step towards finding more tractable fragments.

\printbibliography

\end{document}